\documentclass[runningheads,envcountsame]{llncs}

\usepackage[T1]{fontenc}
\usepackage[english]{babel}

\usepackage{graphicx}

\usepackage{amsfonts, amsmath, amsthm, amssymb, amstext}

\usepackage{hyperref}
\usepackage[capitalize,nameinlink,noabbrev]{cleveref}

\usepackage{mathtools}

\usepackage{mathrsfs}

\usepackage{xfrac}
\usepackage{faktor}

\usepackage[nice]{nicefrac}

\usepackage[shortlabels]{enumitem}

\usepackage{stmaryrd}

\usepackage{xspace}

\usepackage{tikz}
\usetikzlibrary{patterns, arrows, decorations.pathmorphing, decorations.pathreplacing, decorations.markings, decorations.text , 3d, calc, cd, positioning}

\usepackage{dcolumn}

\usepackage{stackrel}

\usepackage{colonequals}

\usepackage[all,cmtip]{xy}

\usepackage{lineno}

\usepackage{CJKutf8} 

\usepackage{csquotes}

\usepackage{chngcntr}

\usepackage{extarrows}

\usepackage{multicol}

\usepackage{float}

\usepackage{lmodern}

\usepackage{comment}

\usepackage{authblk}

\usepackage{bussproofs}

\usepackage{tabularx}

\usepackage{breakcites}

\usepackage{fontawesome}

\usepackage{cases}

\usepackage[old]{old-arrows} 

\usepackage{scalerel}

\usepackage{changepage}
\usepackage{wrapfig}

\usepackage[misc]{ifsym}

\newtheoremstyle{plain}{15pt}{15pt}{\itshape}{}{\bfseries}{.}{.5em}{}
\newtheoremstyle{definition}{15pt}{20pt}{}{}{\bfseries}{.}{.5em}{}

\theoremstyle{plain}

\theoremstyle{definition}
\newtheorem{notation}[theorem]{Notation}

\theoremstyle{remark}

\colorlet{linkcolor}{red!60!black} 
\hypersetup{
    colorlinks=true,
    linkcolor=linkcolor,
    citecolor=linkcolor,
    filecolor=linkcolor,      
    urlcolor=linkcolor,
    pdftitle={Fuzzy Presheaves are Quasitoposes and Consequences in Graph Rewriting},
    pdfpagemode=FullScreen,
}

\newcommand{\customqed}{\hfill $\blacksquare$}

\makeatletter
\def\orcidID#1{\href{http://orcid.org/#1}{\protect\raisebox{-1.25pt}{\protect\includegraphics{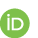}}}}
\makeatother

\DeclareTextFontCommand{\myemph}{\bfseries\em}

\newcommand{\qmarks}[1]{``#1''}

\DeclarePairedDelimiter\set{\{}{\}}
\DeclarePairedDelimiterX\setvbar[2]{\{}{\}}{#1 \nonscript\;\delimsize \vert \nonscript\; #2}
\DeclarePairedDelimiterX\setcolon[2]{\{}{\}}{#1 : #2}

\newcommand{\longequal}{{=\joinrel=}}

\newcommand*{\defeq}{\mathrel{\vcenter{\baselineskip0.5ex \lineskiplimit0pt
                    \hbox{\scriptsize.}\hbox{\scriptsize.}}}%
                    =}

\DeclareMathOperator{\ima}{im}

\newcommand{\inv}{^{\raisebox{.2ex}{$\scriptscriptstyle-1$}}}

\newcommand\quot[2]{{^{\textstyle #1}\big/_{\textstyle #2}}}

\newcommand\restrict[1]{\raisebox{-.6ex}{$|$}_{\scriptscriptstyle #1}}

\newcommand{\N}{\ensuremath{\mathbb{N}}\xspace}

\newcommand{\Q}{\mathbb{Q}}
\newcommand{\R}{\mathbb{R}}

\newcommand{\id}[1]{\mathrm{id}_{#1}}

\let\emptyset\varnothing

\newcommand{\mysetminusD}{\hbox{\tikz[baseline=0]{\draw[-,line width=0.6pt, line cap=round] (3pt, 0) -- (0, 6pt);}}}
\newcommand{\mysetminusT}{\mysetminusD}
\newcommand{\mysetminusS}{\hbox{\tikz[baseline=0]{\draw[-,line width=0.45pt, line cap=round] (2pt, 0) -- (0, 4pt);}}}
\newcommand{\mysetminusSS}{\hbox{\tikz[baseline=0]{\draw[-,line width=0.4pt, line cap=round] (1.5pt, 0) -- (0, 3pt);}}}
\renewcommand{\setminus}{\mathbin{\mathchoice{\mysetminusD}{\mysetminusT}{\mysetminusS}{\mysetminusSS}}}

\renewcommand{\leq}{\leqslant}
\renewcommand{\geq}{\geqslant}
\newcommand*\meet{\wedge}
\newcommand*\join{\vee}
\newcommand*\bigmeet{\bigwedge}

\newcommand{\true}{\ensuremath{\mathsf{True}}\xspace}

\let\lnottemp\lnot
\renewcommand{\lnot}{\lnottemp \hspace*{0.1em}}

\let\oldexists\exists
\let\exists\relax
\newcommand{\exists}{\hspace*{0em}\oldexists\hspace*{0.07em}}
\let\oldforall\forall
\let\forall\relax 
\newcommand{\forall}{\hspace*{0em}\oldforall\hspace*{0.07em}}

\tikzcdset{scale cd/.style={every label/.append style={scale=#1},
    cells={nodes={scale=#1}}}}

\newcommand{\cat}[1]{\mathsf{#1}} 
\newcommand{\Set}{\cat{Set}}
\newcommand{\FinSet}{\cat{FinSet}}

\newcommand{\Graph}{\cat{Graph}}

\newcommand{\labels}{\mathcal{L}}

\newcommand{\elements}[1]{\mathrm{el}(#1)}

\newcommand{\inlsym}{{\mathsf{inl}}}

\newcommand{\inl}{\ensuremath\operatorname{\inlsym}\xspace}

\newcommand{\inclusion}{{\mathsf{incl}}}

\newcommand{\opcat}[1]{{#1}^{\mathrm{op}}}

\newcommand{\FuzzyGraph}{\cat{FuzzyGraph}}

\newcommand{\FuzzySet}[1]{\cat{FuzzySet(#1)}}
\newcommand{\FuzzySetDefault}{\FuzzySet{\labels}}
\newcommand{\FunctorCat}[2]{{#1}^{#2}}
\newcommand{\Presheaf}[1]{\FunctorCat{\Set}{\opcat{#1}}}
\newcommand{\PresheafHat}[1]{\hat{#1}}
\newcommand{\PresheafDefault}{\Presheaf{I}}
\newcommand{\PresheafHatDefault}{\PresheafHat{I}}
\newcommand{\FuzzyPresheaf}[2]{\cat{FuzzyPresheaf}(#1,#2)}
\newcommand{\FuzzyPresheafDefault}{\FuzzyPresheaf{I}{\labels}}

\DeclareMathOperator{\Hom}{Hom}
\DeclareMathOperator{\MonoClass}{Mono}
\newcommand{\mono}{\rightarrowtail}

\newcommand{\regmono}{\hookrightarrow}

\newcommand{\epimono}{
    \rightarrowtail\mathrel{\mspace{-15mu}}\rightarrow
}
\newcommand{\equivcat}{\simeq}
\newcommand{\isom}{\cong}

\newcommand{\Sub}{\mathsf{Sub}}
\newcommand{\pullbackby}[1]{{#1}^\leftarrow}

\newcommand{\refl}{\mathsf{refl}}
\newcommand{\sym}{\mathsf{sym}}

\newcommand{\powerset}{{\mathcal{P}}}

\newcommand{\topology}{\tau}

\makeatletter
\DeclareFontEncoding{LS2}{}{\@noaccents}
\makeatother
\DeclareFontSubstitution{LS2}{stix}{m}{n}

\DeclareSymbolFont{largesymbolsstix}{LS2}{stixex}{m}{n}

\DeclareMathDelimiter{\lParen}{\mathopen}{largesymbolsstix}{"DE}{largesymbolsstix}{"02}
\DeclareMathDelimiter{\rParen}{\mathclose}{largesymbolsstix}{"DF}{largesymbolsstix}{"03}
\DeclareMathDelimiter{\lBrack}{\mathopen}{largesymbolsstix}{"E0}{largesymbolsstix}{"06}
\DeclareMathDelimiter{\rBrack}{\mathclose}{largesymbolsstix}{"E1}{largesymbolsstix}{"07}
\DeclareMathDelimiter{\lBrace}{\mathopen}{largesymbolsstix}{"E8}{largesymbolsstix}{"0E}
\DeclareMathDelimiter{\rBrace}{\mathclose}{largesymbolsstix}{"E9}{largesymbolsstix}{"0F}
\DeclareMathDelimiter{\lbrbrak}{\mathopen} {largesymbolsstix}{"EE}{largesymbolsstix}{"14}
\DeclareMathDelimiter{\rbrbrak}{\mathclose}{largesymbolsstix}{"EF}{largesymbolsstix}{"15}

\DeclareFontFamily{OMX}{MnSymbolE}{}
\DeclareFontShape{OMX}{MnSymbolE}{m}{n}{
    <-6>  MnSymbolE5
   <6-7>  MnSymbolE6
   <7-8>  MnSymbolE7
   <8-9>  MnSymbolE8
   <9-10> MnSymbolE9
  <10-12> MnSymbolE10
  <12->   MnSymbolE12}{}
\DeclareFontShape{OMX}{MnSymbolE}{b}{n}{
    <-6>  MnSymbolE-Bold5
   <6-7>  MnSymbolE-Bold6
   <7-8>  MnSymbolE-Bold7
   <8-9>  MnSymbolE-Bold8
   <9-10> MnSymbolE-Bold9
  <10-12> MnSymbolE-Bold10
  <12->   MnSymbolE-Bold12}{}
 
\DeclareSymbolFont{largesymbolsmnsymbol}  {OMX}{MnSymbolE}{m}{n}

\DeclareMathDelimiter{\langlebar}{\mathopen}{largesymbolsmnsymbol}{'152}{largesymbolsmnsymbol}{'152}
\DeclareMathDelimiter{\ranglebar}{\mathclose}{largesymbolsmnsymbol}{'157}{largesymbolsmnsymbol}{'157}

\makeatletter
\newcommand{\doublearrowtail@inner}[2]{%
  \vcenter{\offinterlineskip
    \halign{%
      ##\cr
      $\m@th#1\rightarrowtail$\cr
      \makebox[\widthof{$\m@th#1\rightarrowtail$}][s]{%
        $\m@th#1\leftarrowtail$%
      }\cr
    }%
  }%
}

\DeclareRobustCommand{\doublearrowtail}{%
  \mathrel{\mathpalette\doublearrowtail@inner\relax}%
}
\makeatother

\makeatletter
\newcommand{\rightleftmapsto@inner}[2]{
  \vcenter{\offinterlineskip
    \halign{
      ##\cr
      $\m@th#1\mapsto$\cr
      \makebox[\widthof{$\m@th#1\mapsto$}][s]{
        $\m@th#1\mapsfrom$
      }\cr
    }
  }
}

\DeclareRobustCommand{\rightleftmapsto}{
  \mathrel{\mathpalette\rightleftmapsto@inner\relax}
}
\makeatother

\makeatletter
\newcommand{\mapstofrom}{\mathpalette\@mapstofrom\relax}
\newcommand*{\@mapstofrom}[2]{%
   \dimen@\fontdimen8
       \ifx#1\displaystyle\textfont\else
       \ifx#1\textstyle\textfont\else
       \ifx#1\scriptstyle\scriptfont\else
       \scriptscriptfont\fi\fi\fi 3
   \mathrel{%
      \vcenter{%
         \vbox{%
            \baselineskip\z@skip
            \lineskip\z@
            \ialign{##\cr$#1\mapstochar\varrightarrow$\cr
            \noalign{\kern\dimen@}%
            $#1\varleftarrow\mapsfromchar$\cr}%
         }%
      }%
   }%
}
\makeatother

\newcommand{\pbpostrong}{PBPO$^{+}$\xspace}

\newenvironment{mysidepicture}[4]{
\nointerlineskip\noindent
\begin{tikzpicture}[overlay]%
  \node at (\textwidth,0) [anchor=north east,rectangle,inner sep=0,outer sep=0,xshift=#2,yshift=#3] {#4};%
\end{tikzpicture}%
\begin{adjustwidth}{0cm}{#1}%
}{%
\end{adjustwidth}%
}

\begin{document}

\title{Fuzzy Presheaves are Quasitoposes}

\subtitle{}

\author{
    Alo\"is Rosset  \orcidID{0000-0002-7841-2318} \textsuperscript{\faEnvelopeO} 
    \and
    Roy Overbeek \orcidID{0000-0003-0569-0947}
    \and 
    J\"{o}rg Endrullis
}

\authorrunning{Alo\"is Rosset, Roy Overbeek and J\"{o}rg Endrullis} %

\institute{
    Vrije Universiteit Amsterdam, Amsterdam, The Netherlands \\
    \email{\{a.rosset, r.overbeek, j.endrullis\}@vu.nl} %
}

\maketitle

\begin{abstract}
    Quasitoposes encompass a wide range of structures, including various categories of graphs.
    They have proven to be a natural setting for reasoning about the metatheory of algebraic graph rewriting.
    In this paper we propose and motivate the notion of \emph{fuzzy presheaves}, which generalises fuzzy sets and fuzzy graphs. 
    We prove that fuzzy presheaves are rm-adhesive quasitoposes, proving our recent conjecture for fuzzy graphs.
    Furthermore, we show that simple fuzzy graphs categories are quasitoposes.
\end{abstract}

\noindent 
The algebraic graph transformation tradition uses category theory to define graph rewriting formalisms.
The categorical approach enables rewriting a large variety of structures,
and allows meta-properties to be studied in uniform ways, such as concurrency, parallelism, termination, and confluence.
Different formalisms have been developed and studied since 1973, such as
DPO~\cite{Ehrig_Pfender_Schneider_1973_DPO},
SPO~\cite{Lowe_1993_SPO},
DPU~\cite{Bauderon_1995_DPU},
SqPO~\cite{Corradini_Heindel_Hermann_Konig_2006_SqPO},
AGREE~\cite{Corradini_Duval_Echahed_Prost_Ribeiro_2015_AGREE},
PBPO~\cite{Corradini_Duval_Echahed_Prost_Ribeiro_2019_PBPO}, and 
\pbpostrong{}~\cite{Overbeek_Endrullis_Rosset_2020_PBPO+}.

The study of meta-properties on a categorical level has given new motivation to study existing concepts, such as quasitoposes~\cite{Wyler_1991}, and has given rise to new ones, such as different notions of adhesivity~\cite{Lack_Sobocinski_Walukiewicz_2004_Adhesive_categories, Lack_Sobocinski_2005_Adhesive_quasiadhesive_categories}.
Quasitoposes are categories with rich structure.
They encompass many examples of particular interest in computer science~\cite{Johnstone_Lack_Sobocinski_2007_Quasitoposes_quasiadhesive_artin_glueing}.
They have moreover been proposed as a natural setting for non-linear rewriting~\cite{Behr_Harmer_Krivine_2021_Concurrency_theorems_quasitoposes}.
In addition, they provide a unifying setting: we have shown that \pbpostrong{} subsumes DPO, SqPO, AGREE and PBPO in quasitoposes~\cite[Theorem 73]{Overbeek_Endrullis_Rosset_2023_PBPO+_Quasitopos}.\footnote{We moreover conjecture that \pbpostrong{} subsumes SPO in quasitoposes~\cite[Rem.~27]{Overbeek_Endrullis_Rosset_2023_PBPO+_Quasitopos}.}

Fuzzy sets generalise the usual notion of sets by allowing elements to have membership values in the unit interval \cite{Zadeh_1965_Fuzzy_sets} or more generally in a poset \cite{Goguen_1967_L_Fuzzy_sets}.
Similarly, fuzzy graphs~\cite{Mori_1995_Fuzzy_graph_rewriting} and fuzzy Petri nets~\cite{Cardoso_Valette_Dubois_1996_Fuzzy_Petri_nets, Srivastava_Tiwari_2011_Fuzzy_petri_nets} have been defined by allowing elements to have membership values.
Having a range of possible membership is used in artificial intelligence for example to express a degree of certainty or strength of a connection.
In graph rewriting, it can be used to restrict matching~\cite{Mori_1995_Fuzzy_graph_rewriting,Parasyuk_2006_Categorical_approach_fuzzy_graph_grammars,Parasyuk_2007_Transformations_of_fuzzy_graphs,Parasyuk_2008_Transformational_approach_fuzzy_graph_models}.
We have recently shown that for fuzzy graphs, where we see the membership value as a label, the order structure on the labelling set lends itself well for implementing relabelling mechanics and type hierarchies~\cite{Overbeek_Endrullis_Rosset_2020_PBPO+,Overbeek_Endrullis_Rosset_2023_PBPO+_Quasitopos}.
These features have proved useful, for example, for  faithfully modelling linear term rewriting with graph rewriting~\cite{Overbeek_2021_Linear_term_rewriting_and_Termination}.

In this paper, we do the following.
\begin{enumerate}
    \item 
    We propose the notion of a \emph{fuzzy presheaf}, which is a presheaf endowed pointwise with a membership function, generalising fuzzy sets and fuzzy graphs.
    
    \item
    We show that fuzzy presheaf categories are rm-adhesive quasitoposes, when the membership values are taken in a complete Heyting algebra (\cref{thm:fuzzy_presheaf_is_a_quasitopos,thm:fuzzy_presheaf_is_rm_adhesive}).
    We obtain as a corollary the known result that fuzzy sets form a quasitopos, prove our conjecture that fuzzy graph form a quasitopos~\cite{Overbeek_Endrullis_Rosset_2023_PBPO+_Quasitopos}, and as a new results that fuzzy sets, graphs, hypergraphs and other fuzzy presheaves are rm-adhesive quasitoposes.
    
    \item
    We examine the related question whether simple fuzzy graphs form a quasitopos. 
    Because simple graphs are not presheaves (nor rm-adhesive), our main theorem cannot be applied.
    However, we show that (directed and undirected) simple fuzzy graphs are quasitoposes by using the notion of separated elements.
\end{enumerate}

An immediate practical application of our result is that the termination technique~\cite{Overbeek_2023_Termination_of_graph_transformation}, recently developed for rm-adhesive quasitoposes is applicable to rewriting fuzzy presheaves. In addition, it follows that our already mentioned subsumption result for \pbpostrong{} holds for fuzzy presheaf categories~\cite[Theorem 73]{Overbeek_Endrullis_Rosset_2023_PBPO+_Quasitopos}.

The paper is structured as follows.
\cref{sec:preliminaries} introduces preliminary definitions of presheaves, Heyting algebra and quasitoposes.
\cref{sec:main_theorem} proves that fuzzy presheaves are quasitoposes.
\cref{sec:rm_adhesive} proves that fuzzy presheaves are rm-adhesive.
\cref{sec:examples} looks at examples and applications.
\cref{sec:simple_fuzzy_graphs} establishes that simple fuzzy graph categories are also quasitoposes. 
\cref{sec:conclusion} discusses future work and concludes.


\section{Preliminaries}
\label{sec:preliminaries}

We assume the reader is familiar with basic notions of category theory
\cite{Awodey_2006, MacLane_1971_Categories_for_the_working_mathematician}.
In this section, we recall basic definitions about fuzzy sets, presheaves, graphs, Heyting algebra and quasitoposes and fix notation concerning the Yoneda embedding, regular monomorphisms and subobjects.

\begin{definition}
    \label{def:contravariant_hom_functor}
    Let $I$ be a locally small category.
    The \myemph{contravariant $\Hom$-functor}
    ${I(-,i) : \opcat{I} \to \Set}$ is defined 
    on objects $j \in \opcat{I}$ as the set of $I$-morphisms $I(j,i)$, and
    on $\opcat{I}$-morphisms $\opcat{\kappa}: j \rightarrow k$ (i.e. $\kappa: k \rightarrow j$ in $I$) as the precomposition
    \(
         I(j,i) \xrightarrow{- \circ \kappa} I(k,i):
         \;
         (j \xrightarrow{\iota} i)
         \ \xmapsto \ \ 
         (k \xrightarrow{\kappa} j \xrightarrow{\iota} i).
    \)
\end{definition}

\begin{definition}
    \label{def:yoneda_embedding}
    Let $I$ be a locally small category.
    The (contravariant) \myemph{Yoneda embedding} $y : I \to \PresheafDefault$ is defined 
    on objects $y(i \in I) = I(-,i)$ as the contravariant $\Hom$-functor, and
    on morphisms $y(j \xrightarrow{\iota} i)$ as the natural transformation that postcomposes by $\iota$,
    i.e., for every $k \in I$:
    \[
        I(k,j) \xrightarrow{\iota \circ -} I(k,i):
        \;
        (k \xrightarrow{\kappa} j)
        \ \xmapsto \ \ 
        (k \xrightarrow{\kappa} j \xrightarrow{\iota} i).
    \]
\end{definition}

\noindent\begin{minipage}{.8\linewidth}
\begin{definition}[{\cite{Zadeh_1965_Fuzzy_sets}}]
    \label{def:fuzzy_set}
    Given a poset $(\labels,\leq)$, an \myemph{$\labels$-fuzzy set} is a pair $(A,\alpha)$ consisting of a set $A$ and a \myemph{membership function} $\alpha:A \to \labels$.
    A morphism $f:(A,\alpha) \to (B,\beta)$ of fuzzy sets is a function $f:A \to B$ such that $\alpha \leq \beta f$, i.e., $\alpha(a) \leq \beta f (a)$ for all $a \in A$.
    They form the category $\FuzzySetDefault$.
    Membership functions are represented in grey.
    We express that $\alpha \leq \beta f$ by the diagram on the right, which we call \myemph{$\leq$-commuting}.
\end{definition}
\end{minipage}
\hfill
\begin{minipage}{.2\linewidth}
\begin{center}
\begin{tikzcd}[column sep=2mm, row sep=5mm,ampersand replacement=\&]
    A \& {} \& B \\
    \& {\color{gray}\labels}
    \arrow["\alpha"', color=gray, from=1-1, to=2-2]
    \arrow["\beta", color=gray, from=1-3, to=2-2]
    \arrow["f", from=1-1, to=1-3]
    \arrow["\leq"{description}, color=gray, draw=none, from=1-2, to=2-2,pos=0.4]
\end{tikzcd}
\end{center}
\end{minipage}

\begin{definition}
    \label{def:presheaf}
    A \myemph{presheaf} on a category $I$ is a functor $F : \opcat{I} \to \Set$.
    A morphism of presheaves is a natural transformation between functors.
    The category of presheaves on $I$ is denoted by $\PresheafDefault$ or $\PresheafHatDefault$. 
\end{definition}

\begin{definition}
    \label{def:graph}
    A \myemph{graph} $A$ consists of a set of vertices $A(V)$, a set of edges $A(E)$, and source and target functions $A(s),A(t) : A(E) \to A(V)$.
    A \myemph{graph homomorphism} $f : A \to B$ is a pair of functions $f_V : A(V) \to B(V)$ and $f_E : A(E) \to B(E)$ respecting sources and targets, i.e., $(B(s),B(t)) \cdot f_E = (f_V \times f_V) \cdot (A(s),A(t))$.
    They form the category $\Graph$, which is in fact the presheaf category on 
    \(
        \opcat{I} = 
        \smash{\begin{tikzcd}[ampersand replacement=\&]
            E
                \ar[r, "s"{description}, shift left=3pt]
                \ar[r, "t"{description}, shift right=3pt]
            \& V
        \end{tikzcd}}
    \).
\end{definition}

\begin{definition}[{\cite{Overbeek_Endrullis_Rosset_2020_PBPO+}}]
    \label{def:fuzzy_graph}
    A \myemph{fuzzy graph} $(A,\alpha)$ consists of a graph $A$ and two membership functions $\alpha_V : A(V) \to \labels(V)$, and $\alpha_E : A(E) \to \labels(E)$, where $\labels(V)$ and $\labels(E)$ are posets.
    A \myemph{fuzzy graph homomorphism} $f:(A,\alpha) \to (B,\beta)$ is a graph homomorphism $f : A \to B$ such that $\alpha_V \leq \beta_V f_V$ and $\alpha_E \leq \beta_E f_E$.
    They form the category $\FuzzyGraph$.
\end{definition}

\begin{remark}
    \label{rem:alternate_definition_of_fuzzy_graph}
    Some authors require that fuzzy graphs satisfy $\labels(V) = \labels(E)$ and $\alpha_E(e) \leq \alpha_V(s(e)) \meet \alpha_V (t(e))$ for all edges $e \in A(E)$ \cite{Rosenfeld_1975_Fuzzy_graphs}.
    For survey works on fuzzy graphs, see e.g.~\cite{Mathew_2018_Fuzzy_graph_theory,Pal_2020_Modern_trends_in_fuzzy_graph_theory}.
\end{remark}

\begin{definition}
    \label{def:regular_monomorphism}
    A monomorphism $m : A \mono B$ is called \myemph{regular}, and denoted $m : A \regmono B$, if it is the equaliser of some  pair of parallel morphisms $B \rightrightarrows C$.
\end{definition}

\noindent\begin{minipage}{.8\linewidth}
\begin{definition}
    \label{def:subobject}
    Monomorphisms into an object $C$ form a class denoted $\MonoClass(C)$.
    There is a preorder on this class: $m \leq m'$ if there exists $u$ such that $m = m'u$. 
    The equivalence relation $m \simeq m' \defeq (m \leq m' \; \meet \; m \geq m')$ is equivalent to having $m = m'u$ for some isomorphism $u$.
    The equivalence classes are the \myemph{subobjects} of $C$. 
    If each representative of a subobject is a special kind of monomorphism, for instance a regular monomorphism, then we talk of \myemph{regular subobjects}. 
\end{definition}
\end{minipage}
\hfill
\begin{minipage}{.2\linewidth}
\begin{center}
\begin{tikzcd}[column sep=2mm, row sep=5mm,ampersand replacement=\&]
    A \&\& B \\
    \& C \& 
    \ar[from=1-1, to=2-2, tail, "m"']
    \ar[from=1-1, to=1-3, dotted, "u"]
    \ar[from=1-3, to=2-2, tail, "m'"]
\end{tikzcd}
\end{center}
\end{minipage}

\medskip

\noindent\begin{minipage}{.8\linewidth}
\begin{definition}[{\cite[Def.~14.1]{Wyler_1991}}]
    \label{def:subobject_classifier}
    Let $\cat{C}$ be a category with finite limits, and $\mathcal{M}$ be a class of monomorphisms.
    An \myemph{$\mathcal{M}$-subobject classifier} in $\cat{C}$ is a monomorphism $\true : 1 \mono \Omega$ such that for every monomorphism $m : A \mono B$ in $\mathcal{M}$, there exists a unique $\chi_A : B \to \Omega$ such that
    the diagram on the right is a pullback square.
    When unspecified, $\mathcal{M}$ is all monomorphisms.
\end{definition}
\end{minipage}
\hfill
\begin{minipage}{.2\linewidth}
\begin{center}
\begin{tikzcd}[column sep=5mm, row sep=5mm,ampersand replacement=\&]
    A \& 1 \\
    B \& \Omega
    \ar[from=1-1, to=1-2, "!"]
    \ar[from=1-1, to=2-1, tail, "m"']
    \ar[from=1-2, to=2-2, "\true"]
    \ar[from=2-1, to=2-2, "\chi_A"']
    \ar[from=1-1, to=2-2, "\lrcorner"{anchor=center, pos=0.125}, draw=none]
\end{tikzcd}
\end{center}
\end{minipage}

\begin{definition}
    \label{def:cartesian_closed_and_locally_cartesian_closed}
    A category $\cat{C}$ is \myemph{cartesian closed} if it has finite products and any two objects $A,B \in \cat{C}$ have an exponential object $B^A \in \cat{C}$.
    In more details, there must exist a functor $-^A : \cat{C} \to \cat{C}$ that is right adjoint to the product functor:
    \(  
        - \times A \dashv -^A.
    \)
    A category $\cat{C}$ is \myemph{locally cartesian closed} if all its slice categories $\cat{C}/C$ are cartesian closed.
\end{definition}

\begin{definition}
    \label{def:Heyting_Algebra}
    We recall some order theory definitions.
    \begin{itemize}[$\bullet$]
        \item 
        A \myemph{lattice} is a poset $(\labels, \leq)$ that has a binary join $\join$ and binary meet $\meet$.

        \item
        A lattice is \myemph{bounded} if it has a greatest element $\top$ and a least element $\bot$.

        \item
        A lattice is \myemph{complete} if all its subsets have both a join and a meet. 

        \item
        A \myemph{Heyting algebra} is a bounded lattice with, for all elements $a$ and $b$, an element $a \Rightarrow b$ such that for all $c$
        \begin{equation} \label{eq:heyting_algebra_cartesian_closed}
            a \meet b \leq c \iff a \leq b \Rightarrow c
        \end{equation}
    \end{itemize}
\end{definition}

We can view a poset/lattice/Heyting algebra as a category where for any two elements $a,b$ there is at most one morphism from $a$ to $b$, denoted $a \leq b$.\footnote{
    See e.g.\ \cite[Def.~6.7]{Awodey_2006}, \cite[Def.~24.1]{Wyler_1991}, or \cite[8.1]{MacLane_Moerdijk_1994}.
}
Meet and join correspond to binary product and coproduct.
Least and greatest elements correspond to initial and a terminal objects.
A bounded lattice is thus finitely bicomplete.
A lattice is complete in the usual sense if and only if it is complete as a category.\footnote{
    Note that completeness and cocompleteness are equivalent for lattices \cite[Def.~6.8]{Awodey_2006}.
}
Heyting algebras are cartesian closed categories because (\ref{eq:heyting_algebra_cartesian_closed}) is the $\Hom$-set adjunction requirement asking $a \Rightarrow b$ to be an exponential object.
The counit of the adjunction tells us that we have \textit{modus ponens} \cite[I.8, Prop.~3]{MacLane_Moerdijk_1994}: 
\begin{equation} \label{eq:heyting_algebra_modus_ponens}
    a \meet (a \Rightarrow b) = a \meet b.
\end{equation}

\begin{definition}
    \label{def:quasitopos}
    A category is a \myemph{quasitopos} if 
    \begin{itemize}
        \item every finite limit and finite colimit exists,
        \item it has a regular-subobject classifier, and
        \item it is locally cartesian closed.
    \end{itemize}
\end{definition}

\section{Fuzzy presheaves are quasitoposes}
\label{sec:main_theorem}

We introduce the notion of a \emph{fuzzy presheaf}, which generalises fuzzy sets and fuzzy graphs.
This allows to obtain in only one definition the additional concepts of fuzzy undirected graphs, fuzzy hypergraphs, fuzzy k-uniform hypergraphs, and plenty more.
For examples see \cref{ex:graph_categories_as_presheaf_categories}.

\begin{definition}[Fuzzy Presheaf]
    \label{def:fuzzy_presheaf}
    Given a category $I$ and a family of posets $(\labels(i),\leq)_{i \in I}$, an \myemph{$\labels$-fuzzy presheaf} is a pair $(A,\alpha)$ consisting of a presheaf $A : \opcat{I} \to \Set$ and a membership function $\alpha : A \to \labels$, or to be more precise, a family of functions $\alpha_i : A(i) \to \labels(i)$ for each $i \in I$.
    A morphism of $\labels$-fuzzy presheaves $f : (A,\alpha) \to (B,\beta)$ is a natural transformation $f :A \to B$ such that $\alpha \leq \beta f$, i.e., $\alpha_i \leq \beta_i f_i$ for all $i \in I$.
    They form a category $\FuzzyPresheafDefault$.
\end{definition}


Note that by taking each $\labels(i)$ to be a singleton set, we retrieve the usual definition of presheaf.

\medskip

\noindent\begin{minipage}{.77\linewidth}
\begin{remark}
    \label{rem:alternate_definition_of_fuzzy_presheaf}
    One can strenghten \cref{def:fuzzy_presheaf} by requiring, for each $j \xrightarrow{\iota} i$ in $\I$, a poset morphism $\labels(\iota) : \labels(i) \to \labels(j)$ such that the diagram on the right $=$-commutes or $\leq$-commutes.
    For
    \(
        \opcat{I} = 
        \smash{\begin{tikzcd}[ampersand replacement=\&]
            E
                \ar[r, "s"{description}, shift left=3pt]
                \ar[r, "t"{description}, shift right=3pt]
            \& V
        \end{tikzcd}}
    \), asking for $\labels(E) = \labels(V)$, for $\labels(s)=\labels(t)=\id{}$, and for $\leq$-commutation gives the fuzzy graphs as in \cref{rem:alternate_definition_of_fuzzy_graph}.
\end{remark}
\end{minipage}
\hfill
\begin{minipage}{.23\linewidth}
\begin{center}
\begin{tikzcd}[ampersand replacement=\&, sep=normal]
    A(i) \& A(j) \\
    {\color{gray}\labels(i)} \& {\color{gray}\labels(j)}
    \ar[from=1-1, to=1-2, "A(\iota)"]
    \ar[from=1-1, to=2-1, color=gray, "\alpha_i"']
    \ar[from=1-2, to=2-2, color=gray, "\alpha_j"]
    \ar[from=2-1, to=2-2, color=gray, "\labels(\iota)"']
    \ar[from=1-1, to=2-2, color=gray, "= \text{ or } \leq"{description}, draw=none]
\end{tikzcd}
\end{center}
\end{minipage}

\medskip

The rest of this section is dedicated to proving the next theorem.

\begin{theorem}
    \label{thm:fuzzy_presheaf_is_a_quasitopos}
    If $I$ is a small category and $(\labels(i))_{i \in I}$ is a family of complete Heyting algebras, then $\FuzzyPresheafDefault$ is a quasitopos.
\end{theorem}

In the rest of this section $I$ is a small category and $\mathcal{L} = \big( \mathcal{L}(i) \big)_{i \in I}$ is a family of complete Heyting algebras.
We give a brief high-level overview of the proof. 
We tackle the properties of being a quasitopos one by one (\cref{def:quasitopos}):
\begin{itemize}
    \item
    \cref{subsec:finite_limits_and_colimits}: 
    Finite limits and colimits exist for presheaves, and it suffices to combine them with a terminal (respectively initial) structure to have finite limits (respectively finite colimits) for fuzzy presheaves.

    \item
    \cref{subsec:regular_subobject_classifier}:
    We modify the subobject classifier for presheaves by giving all its elements full membership.
    This gives us the regular-subobject classifier for fuzzy presheaves.

    \item
    \cref{subsec:cartesian_closed}:
    We combine how exponential objects look like for fuzzy sets and presheaves to obtain exponential objects for fuzzy presheaves, resulting in cartesian closedness.

    \item
    \cref{subsec:locally_cartesian_closed}:
    Lastly, to show that we have a locally cartesian closed category, we show that every slice of a fuzzy presheaf category gives another fuzzy presheaf category, extending on a result from the literature.
\end{itemize}

\subsection{Finite limits and colimits}
\label{subsec:finite_limits_and_colimits}

The goal of this subsection is to prove the following

\begin{proposition}
    \label{prop:fuzzy_presheaf_has_finite_limits_and_colimits}
    $\FuzzyPresheafDefault$ has all finite limits and colimits.
\end{proposition}

\noindent\begin{minipage}{.65\linewidth}
\begin{proof}
    To show existence of finite limits, it suffices to show the existence of a terminal object and of pullbacks~\cite[Prop.\ 2.8.2]{Borceux_1994_vol1}.
    Those are the same as in $\PresheafDefault$, paired with a final structure for the membership function \cite[\textsection3.1.1]{Stout_1993}.
    The terminal object is $1(i) = \set{\; \cdot^\top }$ for each $i \in I$, with the single element having full membership $\top \in \labels(i)$.
    The pullback of two functions $f$ and $g$ is depicted on the right.
    Its membership function must be maximal satisfying $\delta \leq \alpha g'$ and $\delta \leq \beta f'$, and is therefore simply the meet of both conditions.
    The proof that $\FuzzyPresheafDefault$ has all finite colimits is dual.
\end{proof}
\end{minipage}
\hfill
\begin{minipage}{.35\linewidth}
\centering
\begin{tikzcd}[scale cd =.8, column sep=1mm,row sep=2mm]
    && {A \times_C B} \\
    A &&& {} &&&& B \\
    &&&&& C \\
    \\
    \\
    &&&& {\color{gray}\labels}
    \arrow["\alpha"', color=gray, from=2-1, to=6-5]
    \arrow["\gamma"', color=gray, from=3-6, to=6-5]
    \arrow["\beta", color=gray, from=2-8, to=6-5]
    \arrow["\delta"'{pos=0.7}, color=gray, dotted, from=1-3, to=6-5]
    \arrow["{g'}"', dotted, from=1-3, to=2-1]
    \arrow["{f'}", dotted, from=1-3, to=2-8]
    \arrow["f"{pos=0.7}, from=2-1, to=3-6]
    \arrow["g"', from=2-8, to=3-6]
\end{tikzcd}
\\[2ex]
$\delta_i (d) \defeq \alpha_i g'_i (d) \meet \beta_i f'_i (d)$
\end{minipage}


\subsection{Regular-subobject classifier}
\label{subsec:regular_subobject_classifier}

The goal of this subsection is to prove the following.

\begin{proposition}
    \label{prop:fuzzy_presheaf_has_regular_subobject_classifier}
    $\FuzzyPresheafDefault$ has a regular-subobject classifier.
\end{proposition}

\begin{example}
    \label{ex:regular_monomorphisms_in_set_presheaf_fuzzyset}
    We first look at some examples of regular monomorphisms
    \begin{itemize}
        \item
        In $\Set$, a monomorphism is a subset inclusion (up to isomorphism), and all monomorphisms are regular \cite[7.58(1)]{Adamek_1990_JoyOfCat}.
    
        \item
        In $\PresheafDefault$, a morphism $m : A \to B$ is (regular) monic if and only if every $m_i : A(i) \to B(i)$, where $i \in I$, is (regular) monic.
        All monomorphisms are therefore regular for the same reason as in $\Set$.

        \item
        In $\FuzzySetDefault$, a regular monomorphism $m : (A,\alpha) \regmono (B,\beta)$ is an equaliser of some functions $f_1,f_2 : B \rightrightarrows C$.
        Because $m: (A,\beta m) \rightarrow (B,\beta)$ also equalises $f_1$ and $f_2$, the universal property of the equaliser implies $\beta m \leq \alpha \leq \beta m$ and thus equality $a = \beta m$.
        So when $A \subseteq B$, $\alpha$ is the restriction $\beta\restrict{A}$ \cite[3.1.3]{Stout_1993}.
        \[\begin{tikzcd}[column sep = 2em, row sep = .4em] 
        	A && B && C \\
        	A \\
        	&& {\color{gray}\labels}
        	\arrow["{f_1}", shift left=1, from=1-3, to=1-5]
        	\arrow["{f_2}"', shift right=1, from=1-3, to=1-5]
        	\arrow["m", hook, from=1-1, to=1-3]
        	\arrow[equal, from=2-1, to=1-1]
        	\arrow["\beta", color=gray, from=1-3, to=3-3]
        	\arrow["\gamma", color=gray, from=1-5, to=3-3]
        	\arrow["\beta m"', color=gray, from=2-1, to=3-3]
        	\arrow["\alpha", color=gray, pos=.6, from=1-1, to=3-3]
                \arrow["m", hook, crossing over, from=2-1, to=1-3]
        \end{tikzcd}\]
    \end{itemize}
\end{example}

Fuzzy presheaves being pointwise fuzzy sets, we have the following.

\medskip

\noindent\begin{minipage}{.75\linewidth}
\begin{lemma}
    A regular monomorphism ${m : (A,\alpha) \regmono (B,\beta)}$ between fuzzy presheaves is a natural transformation ${m : A \to B}$ such that each $m_i : (A(i),\alpha_i) \regmono (B(i),\beta_i)$ is regular, i.e., $\alpha_i = \beta_i m_i$.
\end{lemma}
\end{minipage}
\hfill
\begin{minipage}{.25\linewidth}
\centering
\begin{tikzcd}[column sep=0mm, row sep=3mm, ampersand replacement=\&]
    A(i) \& {} \& B(i) \\
    \& {\color{gray}\labels(i)}
    \arrow["{\alpha_i}"', color=gray, from=1-1, to=2-2]
    \arrow["{\beta_i}", color=gray, from=1-3, to=2-2]
    \arrow["{m_i}", from=1-1, to=1-3]
        \arrow["=" description, draw=none, color=gray, from=1-2, to=2-2]
\end{tikzcd}
\end{minipage}

\begin{notation}
    \label{not:pullback_along}
    We denote the pullback of a morphism $m : A \to C$ along another morphism $f : B \to C$ by $\pullbackby{f}m$.
\end{notation}

\begin{example}
    \label{ex:subobject_classifier_in_set_fuzzyset_graph_presheaf}
    We continue \cref{ex:regular_monomorphisms_in_set_presheaf_fuzzyset} and look at the (regular-)subobject classifiers of the aforementioned categories. 
    \begin{itemize}
        \item 
        In $\Set$, we have $\true : \set{\cdot} \to \Omega \defeq \set{0,1} : \cdot \mapsto 1$.
        For $m : A \mono B$, the characteristic function $\chi_A : B \to \set{0,1}$ is $1$ on $m(A) \subseteq B$ and $0$ on the rest.
    \end{itemize}
    
    \begin{mysidepicture}{4cm}{0.5ex}{0ex}{%
    \begin{tikzcd}[row sep=2mm, column sep=5mm,ampersand replacement=\&]
        A \&\& {\set{\cdot}} \\
        \& {\color{gray} \labels} \\
        B \&\& {\set{0,1}}
        \arrow["m"', tail, from=1-1, to=3-1]
        \arrow["\true", from=1-3, to=3-3]
        \arrow["{!}", from=1-1, to=1-3]
        \arrow["{\chi_A}"', from=3-1, to=3-3]
        \arrow["{\alpha=\beta m}"{description,outer sep=0,inner sep=0,pos=0.4}, color=gray, from=1-1, to=2-2]
        \arrow["\top"{description}, color=gray, from=1-3, to=2-2]
        \arrow["\top"{description}, color=gray, from=3-3, to=2-2]
        \arrow["\beta"{description}, color=gray, from=3-1, to=2-2]
    \end{tikzcd}
    }
    \begin{itemize}[topsep=-2ex]
        \item
        In $\FuzzySetDefault$ \cite[3.1.3]{Stout_1993}, we have $\Omega \defeq \set{\; 0^\top,1^\top}$ with both elements having full membership.
        The square being a pullback forces $\alpha = \beta m$.
        Hence $\Omega$ classifies only regular subobjects.
    \end{itemize}
    \end{mysidepicture}
    \smallskip

    \begin{mysidepicture}{4cm}{-4ex}{-2ex}{%
        \begin{tikzcd}[ampersand replacement=\&]
            0
                \arrow[loop, distance=1em, in=205, out=155, "0"']
                \arrow[r, "t"', bend right=15]
            \& 1
                \arrow[loop, distance=1em, in=295, out=245, "{(s, t)}"']
                \arrow[loop, distance=1em, in=115, out=65, "s \to t"']
                \arrow[l, "s"', bend right=15]
        \end{tikzcd}
    }
    \begin{itemize}
        \item
        In $\Graph$
        ~\cite{vigna2003graphistopos}, the terminal object is
        \(
            \set{ 
            \begin{tikzcd}
                \cdot \arrow[loop, distance=.8em, in=200, out=160]
            \end{tikzcd}
            }
        \), 
        the classifying object $\Omega$ is shown on the right, and
        \smash{\(
            \true :
            \begin{tikzcd}
                \cdot 
                \arrow[loop, distance=.8em, in=200, out=160]
            \end{tikzcd}
            \mapsto 
            \begin{tikzcd}
                1
                \arrow[loop, distance=.8em, in=20, out=-20, "s \to t"'] 
            \end{tikzcd}
        \)}.
        Given a graph $H$ and a subgraph $G \subseteq H$,  the characteristic function $\chi_G$ sends vertices in $G$ to $1$ and edges in $G$ to $s \to t$. \pagebreak
    \end{itemize}
    \end{mysidepicture}
    \begin{itemize}[topsep=-2ex]
        \item []
        Vertices not in $G$ are sent to $0$; and edges not in $G$ are sent to $(s,t)$ if both endpoints are in $G$, to $0$ if neither endpoint is in $G$, and to $s$ or $t$ if only the source or target is in $G$, respectively.
        
    \end{itemize}
\end{example}

When $I$ is small, $\PresheafDefault$ has a subobject classifier \cite{MacLane_Moerdijk_1994}.
By smallness of $I$, the class of the subobjects of a presheaf $B \in \PresheafHatDefault$ (also called subpresheaves) is necessarily a set.
This allows us to define a functor $\Sub_{\PresheafHatDefault} : \opcat{\PresheafHatDefault} \to \Set$:
\begin{itemize}[topsep=2pt]
    \item
    $\Sub_{\PresheafHatDefault}(B) \defeq \set{\text{subobjects } m : A \mono B}$,

    \item
    $\Sub_{\PresheafHatDefault}(f:B \to C)$ sends a subobject $(A\stackrel{n}{\mono} C)$ of $C$ to $\pullbackby{f}n$, the pullback of $n$ along $f$, which is a subobject of $B$.
\end{itemize}
The classifying object is the presheaf $\Omega \defeq \Sub_{\PresheafHatDefault} (y(-)) : \opcat{I} \to \Set$.
For $i \in I$, $\Omega(i)$ is thus the set of all subpresheaves of $y(i)$, which are sometimes called \textit{sieves}.
Let $\true_i : \set{\cdot} \to \Omega(i) : \cdot \mapsto y(i)$, because $y(i)$ is indeed a subpresheaf of itself.
We prove that $\true$ classifies subobjects.
Pick a subpresheaf $m : A \mono B$, and suppose w.l.o.g.~$A(i) \subseteq B(i)$ for each $i \in I$.
We want to define the natural transformation $\chi_A : B \to \Omega$.
For $i \in I$ and $b \in B(i)$, then $(\chi_A)_i(b) \in \Omega(i)$ must be a subpresheaf of $y(i) = I(-,i)$.
On $j \in I$, let
\begin{equation}
    \label{eq:def_chi_A}
    (\chi_A)_i(b)(j)
    \defeq \setvbar{\iota : j \to i}{B(\iota)(b) \in A_j}.
\end{equation}

\begin{lemma} \label{lem:chi_A_1)is_natural_2)makes_a_pullback_square_3)is_unique}
    We have the following properties:
    \begin{enumerate}
        \item 
        $\chi_A$ is natural,
        
        \item
        \label{item:chi_A_makes_pullback_square}
        $\chi_A$ characterises the subpresheaf $A \subseteq B$, i.e., gives a pullback square, and
        
        \item
        $\chi_A$ is unique in satisfying \cref{item:chi_A_makes_pullback_square}.
        \customqed\footnote{The symbol $\blacksquare$ denotes that the proof is in the appendix.
        A sketch of the proof is in \cite[p.37-38]{MacLane_Moerdijk_1994}.
        }
    \end{enumerate}
\end{lemma}

In $\FuzzyPresheafDefault$, we consider the same construction and give full membership value to $\Omega$, analogously to the fuzzy set regular-subobject classifier.
Explicitly, for $i \in I$, let $\omega_i : \Omega(i) \to \labels(i) : \big(A \mono y(i)\big) \mapsto \top$.
By \cref{lem:chi_A_1)is_natural_2)makes_a_pullback_square_3)is_unique}, $(\Omega,\omega)$ is the regular-subobject classifier that we seek and \cref{prop:fuzzy_presheaf_has_regular_subobject_classifier} is hence proved.

\subsection{Cartesian closed}
\label{subsec:cartesian_closed}

The goal of this section is to prove the following.

\begin{proposition} \label{prop:cartesian_closed}
    $\FuzzyPresheafDefault$ is cartesian closed.
\end{proposition}

Let $(A,\alpha)$ be a fuzzy presheaf.
We construct a product functor, an exponentiation functor, and an adjunction between them:
$- \times (A,\alpha) \dashv -^{(A,\alpha)}$.


\medskip

\noindent\begin{minipage}{.65\linewidth}
\begin{definition}
    \label{def:product_functor}
    The product of two fuzzy presheaves $(C,\gamma)$ and $(A,\alpha)$ is $(C \times A, \gamma\pi_1 \meet \alpha\pi_2)$.
    The membership function is written $\gamma \meet \alpha$ for short.
    Given ${f:(C',\gamma') \to (C,\gamma)}$, then $f \times (A,\alpha) \defeq f \times \id{A}$.
\end{definition}
\end{minipage}
\hfill
\begin{minipage}{.35\linewidth}
\centering
\begin{tikzcd}[sep=tiny, ampersand replacement=\&]
    C' \times A
        \ar[rr, "f \times \id{A}"]
        \ar[dr, color=gray, "\gamma' \meet \alpha"']
    \& {}
        \ar[d, "\leq" description, draw=none, color=gray]
    \& C \times A
        \ar[dl, color=gray, "\gamma \meet \alpha"] 
    \\
    \& {\color{gray}\labels} 
\end{tikzcd}
\end{minipage}

\smallskip

\begin{lemma}
    The $- \times (A,\alpha)$ functor is well-defined on morphisms.
\end{lemma}

\begin{proof}
    Because $\gamma' \leq \gamma f$, it follows that $\gamma' \meet \alpha \leq (\gamma \meet \alpha) \circ (f \times \id{A})$.
\end{proof}

\begin{example} \label{ex:exponential_objects_in_set_fuzzyset_presheaf}
    To construct exponential objects    of fuzzy presheaves, we first look at simpler examples.
    \begin{itemize}
        \item 
        In $\Set$, given two sets $A,B$, then $B^A \defeq \Set(A,B)$.

        \item
        In $\FuzzySet{\labels}$ \cite[p.~81]{Stout_1993}, given fuzzy sets $(A,\alpha: A \to \labels), (B, \beta: B \to \labels)$, then $(B,\beta)^{(A,\alpha)}$ has carrier $B^A$ as in $\Set$ and membership function
        \[
            \theta(f: A \to B) \defeq \bigmeet_{a \in A} \left( \alpha(a) \Rightarrow \beta f (a) \right).
        \]
        Logically speaking, the truth value of $f$ is the minimum truth-value of \qmarks{truth preservation by $f$}, i.e., \qmarks{if $a$ then $f(a)$}.
        
        \item
        In $\Graph$ \cite{vigna2003graphistopos} (i.e., directed multigraphs), given two graphs $A,B$ then $B^A$ is the graph
        \begin{itemize}[-]
            \item
            with vertex set $B^A(V) \defeq \Graph \big(\set{\cdot} \times A, B\big)$,

            \item
            with edge set $B^A(E) \defeq \Graph \big(\set{s \to t} \times A,B \big)$,

            \item
            the source of $m: \set{s \to t} \times A \to B$ is $m(s,-) : \set{\cdot} \times A \to B$, and
            
            \item
            the target function is analogous.
        \end{itemize}

        \item
        In $\PresheafHatDefault$ \cite[Section 8.7]{Awodey_2006}, given $A,B : \opcat{I} \to \Set$, then $B^A=\hat{I}(y(-)\times A, B)$ is: 
        \begin{itemize}[-,topsep=2pt]
            \item 
            on $i \in I$, the set of all natural transformations $y(i) \times A \Rightarrow B$.

            \item
            on $\iota : j \to i$, the precomposition by $y(\iota) \times \id{A}$.
            \begin{equation}
                \label{eq:exponent_def_objects_iota}
                y(i) \times A \xrightarrow{m} B
                \quad \longmapsto \quad 
                y(j) \times A \xrightarrow{y(\iota) \times \id{A}} y(i) \times A \xrightarrow{m} B
            \end{equation}
        \end{itemize}
        Notice how this generalises the $\Graph$ case: for 
        \(
            \opcat{I} = 
            \smash{\begin{tikzcd}[ampersand replacement=\&]
                E
                    \ar[r, "s"{description}, shift left=3pt]
                    \ar[r, "t"{description}, shift right=3pt]
                \& V
            \end{tikzcd}}
        \),
        we indeed have $y(V) \isom \set{\cdot}$ and $y(E) \isom \set{s \to t}$.
    \end{itemize}
\end{example}

\begin{definition}
    \label{def:exponent_functor}
    Given fuzzy presheaves $(A,\alpha)$ and $(B,\beta)$, let
    \( 
        (B,\beta)^{(A,\alpha)}
        \defeq
        \big( \PresheafHatDefault(y(-) \times A,B), \theta \big).
    \)
    The carrier is the same as in $\PresheafHatDefault$ above.
    On $i \in I$, it is the set of natural transformations $y(i) \times A \Rightarrow B$ and on $\iota:j \to i$ it is the precomposition by $y(\iota) \times \id{A}$.
    The membership function is, for each $i \in I$:
    \begin{equation}
        \label{eq:exponent_def_objects_i}
        \theta_i 
        \Big(
            y(i) \times A \xrightarrow{m} B
        \Big)
        \defeq
        \bigmeet_{a \in A(i)}
        \big(
            \alpha_i a \Rightarrow \beta_i m_i(\id{i},a)
        \big).
    \end{equation}
    Given $(B,\beta) \stackrel{g}{\to} (B', \beta')$, the morphism $g^{(A,\alpha)}$ is postcomposition by $g$:
    \begin{equation} \label{eq:exponent_def_morphisms}
        \begin{tikzcd}[row sep=tiny, column sep = small]
            \PresheafHatDefault(y(i) \times A,B)
                \ar[rr, "g^{(A,\alpha)}_i \defeq g \circ -"]
                \ar[dr, color=gray, "\theta_i"']
            && \PresheafHatDefault(y(i) \times A,B')
                \ar[dl, color=gray, "\theta'_i"] \\
            & {\color{gray}\labels(i)} & 
        \end{tikzcd}
    \end{equation}
\end{definition}

\pagebreak

\begin{lemma}  
    The exponentiation functor is well-defined on morphisms.
\end{lemma}

\begin{proof}
    The mapping $\smash{g^{(A,\alpha)}}$ is clearly natural in $i \in I$.
    Indeed, $\smash{g_i^{(A,\alpha)}}$ is postcomposition by $g$ (\ref{eq:exponent_def_morphisms}), $\smash{\PresheafHatDefault(y(\iota)\times A, B)}$ is precomposition by $y(\iota)\times\id{A}$ (\ref{eq:exponent_def_objects_iota}), and pre- and postcomposition always commute.
    
    Fix an arbitrary $i \in I$.
    The other criteria to verify is that (\ref{eq:exponent_def_morphisms}) $\leq$-commutes, i.e., $\smash{\theta_i \leq \theta'_i g_i^{(A,\alpha)}}$.
    Take some $m : y(i) \times A \to B$.
    We prove $\smash{\theta_i(m) \leq \theta'_i  g_i^{(A,\alpha)} (m)}$.
    By hypothesis $g:(B,\beta) \rightarrow (B', \beta')$, so $\beta_i \leq \beta'_i g_i$.
    For any $a' \in A(i)$, we have $m_i(\id{i},a') \in B(i)$.
    \begin{align*}
        &\Rightarrow
        \forall a' \in A(i): \beta_i m_i (\id{i},a') \leq \beta_i' g_i m_i (\id{i},a') 
        \\
        &\Rightarrow 
        \forall a' \in A(i): \alpha_i(a') \meet \beta_i m_i (\id{i},a') \leq \beta_i' g_i m_i (\id{i},a') 
        \\
        &\stackrel{\smash{\eqref{eq:heyting_algebra_modus_ponens}}}{\Rightarrow}
        \forall a' \in A(i): \alpha_i(a') \meet (\alpha_i(a') \Rightarrow \beta_i m_i (\id{i},a')) \leq \beta'_i g_i m_i (\id{i},a') 
        \\
        &\stackrel{\smash{\eqref{eq:heyting_algebra_cartesian_closed}}}{\Rightarrow}
        \forall a' \in A(i): (\alpha_i (a') \Rightarrow \beta_i m_i (\id{i},a')) \leq (\alpha_i (a') \Rightarrow \beta'_i g_i m_i (\id{i},a'))
        \\
        &\Rightarrow
        \forall a' \in A(i): \textstyle{\bigmeet}_{a \in A(i)} (\alpha_i (a) \Rightarrow \beta_i m_i (\id{i},a)) \leq (\alpha_i (a') \Rightarrow \beta'_i g_i m_i (\id{i},a')) 
        \\
        &\Rightarrow
        \textstyle{\bigmeet}_{a \in A(i)} (\alpha_i (a) \Rightarrow \beta_i m_i (\id{i},a)) \leq \textstyle{\bigmeet}_{a' \in A(i)} (\alpha_i (a') \Rightarrow \beta'_i g_i m_i (\id{i},a')) 
        \\
        &\stackrel{\smash{\mathclap{ \eqref{eq:exponent_def_objects_i}, \eqref{eq:exponent_def_morphisms} }}}{\Rightarrow} 
        \hspace*{1em}
        \theta_i (m) \leq \theta'_i g^{(A,\alpha)}_i (m). 
        \qedhere
    \end{align*}
\end{proof}

\begin{lemma}
    \label{lem:fuzzy_presheaf_exponential_adjunction}
    Given a fuzzy presheaf $(A,\alpha)$, then $- \times (A,\alpha) \dashv -^{(A,\alpha)}$.
    \customqed
\end{lemma}

As a consequence, \cref{prop:cartesian_closed} is now proved.

\subsection{Locally cartesian closed}
\label{subsec:locally_cartesian_closed}

The goal of this section is to prove the following.

\begin{proposition} \label{prop:locally_cartesian_closed}
    $\FuzzyPresheafDefault$ is locally cartesian closed.
\end{proposition}

Because we consider slice categories, let us fix an arbitrary fuzzy presheaf $(D,\delta)$ for the rest of this section.
We aim to prove that $\nicefrac{\FuzzyPresheafDefault}{(D,\delta)}$ is cartesian closed.
Here are the details of this category.
\begin{itemize}[$\bullet$]
    \item 
    Its objects are of the form $((A,\alpha),p)$ where $(A,\alpha)$ is a fuzzy presheaf and $p : (A,\alpha) \to (D,\delta)$ is a fuzzy presheaf morphism to the fixed object $(D,\delta)$, hence satisfying $\alpha \leq \delta p$.
    We denote this object by $(A,\alpha,p)$ for simplicity.

    \item
    A morphism $f:(A,\alpha,p) \to (B,\beta,q)$ is a natural transformations $f: A \to B$ satisfying $p = qf$.
    For $p,q$ and $f$ to be also well-defined in the base category $\FuzzyPresheafDefault$, we must also have $\alpha \leq \delta p$, $\beta \leq \delta q$ and $\alpha \leq \beta f$.
\end{itemize}
\[
    \begin{tikzcd}[sep=small]
        A
            \ar[rr, "p"]
            \ar[dr, color=gray, "\alpha"']
        & {}
            \ar[d, "\leq" description, draw=none, color=gray]
        & D
            \ar[dl, color=gray, "\delta"] \\
        & {\color{gray}\labels} & 
    \end{tikzcd}
    \qquad \qquad
    \begin{tikzcd}[sep=tiny]
        &&& D \\
        A &&&& B \\
        \\
        && {\color{gray}\labels}
        \arrow["\delta"'{pos=0.9}, color=gray, from=1-4, to=4-3]
        \arrow["f", crossing over, from=2-1, to=2-5]
        \arrow["p", from=2-1, to=1-4]
        \arrow["q"', from=2-5, to=1-4]
        \arrow["\alpha"', color=gray, from=2-1, to=4-3]
        \arrow["\beta", color=gray, from=2-5, to=4-3]
    \end{tikzcd}
\]
To prove locally cartesian closedness, we extend on the technique
used in \cite[Lemma 9.23]{Awodey_2006}.
We demonstrate that the slice category $\nicefrac{\FuzzyPresheafDefault}{(D,\delta)}$ is equivalent to \textit{another} category of fuzzy presheaves $\FuzzyPresheaf{J}{\tilde{\labels}}$ for some other $J$ and $\tilde{\mathcal{L}}$.
Because cartesian closedness is preserved by categorical equivalence \cite[Ex.~I.4]{MacLane_Moerdijk_1994}, and because $\FuzzyPresheaf{J}{\tilde{\labels}}$ is cartesian closed by \cref{prop:cartesian_closed}, then \cref{prop:locally_cartesian_closed} follows immediately.
The category $J$ is the category of elements of $D$, so let us recall its definition.

\begin{definition}[{\cite[p.\ 203]{Awodey_2006}}]
    \label{def:category_of_elements}
    Given a presheaf $D \in \PresheafDefault$, its \myemph{category of elements} is denoted $\elements{D}$ or $\int_I D$ and contains the following:
    \begin{itemize}
        \item 
        objects are pairs $(i,d)$ where $i \in I$ and $d \in D(i)$, and
        
        \item
        morphisms $\iota: (j,e) \to (i,d)$ are $I$-morphisms $\iota:j \to i$ such that $D(\iota)(d) = e$.
    \end{itemize}
\end{definition}

\begin{example}
    \label{ex:category_of_elements_of_graph}
    Recall that $\Graph$ it is the presheaf category on 
    $\opcat{I} = \smash{
        \begin{tikzcd}[ampersand replacement=\&]
            E
                \ar[r, "s"{description}, shift left=3pt]
                \ar[r, "t"{description}, shift right=3pt]
            \& V
        \end{tikzcd}
    }
    $.
    Given a graph $D$, the objects of its category of elements are ${\setvbar{(V,v)}{v \in D(V)}} \cup {\setvbar{(E,e)}{e \in D(E)}}$.
    Each edge $e \in D(E)$ induces two $\elements{D}$-morphisms, one going onto the source of $e$, $(E,e) \to (V,s(e))$, and one onto its target $(E,e) \to (V,t(e))$.
    All morphisms of $\elements{D}$ are obtained this way.
\end{example}

\begin{lemma} \label{lem:equivalence_of_categories}
    Let $I$ be a small category and $(\labels(i))_{i \in I}$ a family of complete Heyting algebras.
    For any object $(D,\delta) \in \FuzzyPresheafDefault$, we have an equivalence of categories
    \[
        \quot{\FuzzyPresheafDefault}{(D,\delta)} 
        \equivcat
        \FuzzyPresheaf{\elements{D}}{\tilde{\labels}}
    \]
    where for $(i,d) \in \opcat{\elements{D}}$, $\tilde{\labels}(i,d) \defeq \labels(i)_{\leq \delta_i(d)}$.
    \customqed
\end{lemma}

By \cref{lem:equivalence_of_categories,prop:cartesian_closed}, then \cref{prop:locally_cartesian_closed} is proved.
By \cref{prop:fuzzy_presheaf_has_finite_limits_and_colimits,prop:fuzzy_presheaf_has_regular_subobject_classifier,prop:locally_cartesian_closed}, the proof of \cref{thm:fuzzy_presheaf_is_a_quasitopos} is complete.

\begin{remark}
    If $I$ is a finite category, then every finite presheaf category $\FinSet^{\opcat{I}}$ is a topos~\cite[Ex.~5.2.7]{Borceux_1994_vol3}.
    Analogously, if $I$ is finite we also have that every finite fuzzy presheaf category is a quasitopos.
\end{remark}

\section{Fuzzy presheaves are rm-adhesive}
\label{sec:rm_adhesive}

In this section we show that fuzzy presheaves are rm-adhesive. To state and motivate the result we need some definitions.

\smallskip
\begin{mysidepicture}{3.2cm}{0.5ex}{-2ex}{%
    \begin{tikzcd}[row sep={20,between origins},column sep={20,between origins}, ampersand replacement=\&]
        \& F  \ar[rr] \ar[dd] \ar[dl] \& \&  G \ar[dd] \ar[dl] \\
        E \ar[rr, crossing over] \ar[dd] \& \& H \\
          \& B \ar[rr] \ar[dl] \& \&  C \ar[dl] \\
        A \ar[rr] \&\& D \ar[from=uu,crossing over]
    \end{tikzcd}
}
\begin{definition}[{\cite{Lack_Sobocinski_Walukiewicz_2004_Adhesive_categories}}]
    A pushout square is \myemph{Van Kampen (VK)} if, whenever it lies at the bottom of a commutative cube, like $ABCD$ on the diagram, where the back faces $FBAE$, $FBCG$ are pullbacks, then
    \begin{center}
        the front faces are pullbacks $\iff$ the top face is pushout.
    \end{center}
    A pushout square is \myemph{stable} (under pullback) if, whenever it lies at the bottom of such a cube, then
    \begin{center}
        the front faces are pullbacks $\implies$ the top face is pushout.
    \end{center}
\end{definition}
\end{mysidepicture}

\begin{definition}[{\cite{Garner_Lack_2012_On_the_axioms_for_adhesive_and_quasiadhesive_categories}}]
    A category $\cat{C}$ is
    \begin{itemize}
        \item \myemph{adhesive} if pushouts along monomorphisms exist and are VK;
        \item \myemph{rm-adhesive}\footnote{Also known as \myemph{quasiadhesive}.} if pushouts along regular monomorphisms exist and are VK;
        \item \myemph{rm-quasiadhesive} if pushouts along regular monomorphisms exist, are stable, and are pullbacks.
    \end{itemize}
\end{definition}

Toposes are adhesive~\cite[Proposition 9]{Lack_Sobocinski_Walukiewicz_2004_Adhesive_categories} (and thus rm-adhesive), and quasitoposes are rm-quasiadhesive~\cite{Garner_Lack_2012_On_the_axioms_for_adhesive_and_quasiadhesive_categories}. Not all quasitoposes are rm-adhesive. The category of simple graphs is a counterexample~\cite[Corollary 20]{Johnstone_Lack_Sobocinski_2007_Quasitoposes_quasiadhesive_artin_glueing}.

\smallskip

\begin{mysidepicture}{3.7cm}{0.5ex}{-2ex}{%
    \begin{tikzcd}[column sep=12, row sep=12,ampersand replacement=\&,nodes={rectangle,inner sep=0.5mm,outer sep=0}]
        \& K
            \arrow[ld, "u" description]
            \arrow[rrd, "r" description, pos=0.3]
            \arrow[dd, equals]
        \& \& \\
        G_K
            \arrow[dd, "u'" description]
        \& \& \& R
            \arrow[ld, "w" description] \arrow[dd, equals]
        \\
        \& K
            \arrow[ld, "t_K" description]
            \arrow[rrd, "r" description, pos=0.3]
        \& G_R
            \arrow[from=llu, "g_R" description, crossing over, pos=0.25]
        \& \\
        K'
            \arrow[rrd, "r'" description]
        \& \& \& R 
            \arrow[ld, "t_R" description]
        \\
        \& \& R'
            \arrow[from=uu, "w'" description, crossing over]
        \&                                                          
    \end{tikzcd}
}%
\begin{remark}
    Knowing whether a quasitopos $\cat{C}$ is rm-adhesive is considerably relevant for \pbpostrong{}. Every \pbpostrong{} rewrite step defines a commutative cube (shown on the right)
    where the bottom and top faces are pushouts, and the back faces are pullbacks. So if $t_K$ is a regular mono (as often is the case), it follows that the front faces are pullbacks. The termination method for \pbpostrong{} rewriting~\cite{Overbeek_2023_Termination_of_graph_transformation}, recently developed by the second and third author, depends on the front faces being pullbacks, as well as other quasitopos properties. We expect that the dependency will hold for future theory as well.
\end{remark}
\end{mysidepicture}
\medskip

\emph{Coproducts} are a generalisation of the notion of disjoint unions in $\Set$.
However, a notion of \emph{union} can still be defined.
In a quasitopos, there is an elegant description of binary union of regular subobjects \cite[Proposition 10(iii)]{Johnstone_Lack_Sobocinski_2007_Quasitoposes_quasiadhesive_artin_glueing}.
\begin{definition}
    \label{def:binary_union_of_regular_subobjects_in_quasitopos}
    The \emph{binary union} of two regular subobjects $f: A \regmono C$ and $g : B \regmono C$ in a quasitopos is obtained as the pushout of the pullback of $f$ and $g$.
    \begin{equation}
        \label{eq:binary:union:diagram}
        \begin{tikzcd}[column sep=normal, row sep=tiny]
            & A
                \arrow[rd, "f''" description]
                \arrow[rrrd, "f" description, hook, bend left=10]
            & & \\
            D \defeq A \sqcap_C B 
                \arrow[ru, "g'" description]
                \arrow[rd, "f'" description]
            & & A \sqcup_{D} B
                \arrow[rr, "h" description, dotted]
            && C \\
            & B
                \arrow[rrru, "g" description, hook, bend right=10]
                \arrow[ru, "g''" description]
            & &  
        \end{tikzcd}
    \end{equation}
\end{definition}

\begin{proposition}[{\cite[Proposition 2.4]{Garner_Lack_2012_On_the_axioms_for_adhesive_and_quasiadhesive_categories}, \cite[Proposition 3]{Behr_2021_Concurrency_theorems}}]
    \label{proposition:binary:union:monic}
    In a quasitopos, if $f$ or $g$ of~\eqref{eq:binary:union:diagram} is regular monic, then $h$ is monic. \qed
\end{proposition}


To prove  $\FuzzyPresheafDefault$ rm-adhesive, we use the following theorem.

\begin{theorem}[{\cite[Theorem 21]{Johnstone_Lack_Sobocinski_2007_Quasitoposes_quasiadhesive_artin_glueing}}]
    \label{theorem:rm-adhesive:iff:union}
    If $\cat{C}$ is a quasitopos, then $\cat{C}$ is rm-adhesive iff the class of regular subobjects is closed under binary union.
    \qed
\end{theorem}


\begin{theorem}
    \label{thm:fuzzy_presheaf_is_rm_adhesive}
    If $I$ is a small category and $(\labels(i))_{i \in I}$ is a family of complete Heyting algebras, then $\FuzzyPresheafDefault$ is rm-adhesive.
\end{theorem}
\begin{proof}
    By \cref{thm:fuzzy_presheaf_is_a_quasitopos}, $\FuzzyPresheafDefault$ is a quasitopos, and so by \cref{theorem:rm-adhesive:iff:union}, it suffices to show that the class of regular subobjects is closed under binary union. By the reasoning given in Example~\ref{ex:regular_monomorphisms_in_set_presheaf_fuzzyset}, it suffices to show this property pointwise, i.e., on the level of $\FuzzySetDefault$.
    
    So let two regular subobjects $f : A \regmono C$ and $g : B \regmono C$ of $\FuzzySetDefault$ be given.
    In~\eqref{eq:binary:union:diagram}, $f'$ and $g'$ are regular monic by (general) pullback stability, and then $f''$ and $g''$ are regular monic by pushout stability in quasitoposes~\cite[Lemma A.2.6.2]{Johnstone_elephant_2002}. 
    Moreover, by \cref{proposition:binary:union:monic}, $h$ is monic. 
    
    Now consider an arbitrary element $y \in A \sqcup_D B$. 
    By general pushout reasoning for $\Set$, $y$ is in the image of $f''$ or $g''$, say w.l.o.g.~$y = f''(a)$ for some $a \in A$.
    Because $f = hf''$ and $f$ is regular (i.e., membership-preserving), $a$ and $h(y)$ have the same membership value. 
    The membership value of $y$ is sandwiched between the ones of $a$ and of $h(y)$ and all three are thus equal.
    This means that $h$ preserves the membership, i.e., is regular.
\end{proof}

\section{Examples and applications}
\label{sec:examples}

In this section, we showcase the wide variety of structures that presheaf categories cover.
We can therefore add labels from a complete Heyting algebra to all of those categories, and have an rm-adhesive quasitopos.

\begin{example}
    \label{ex:graph_categories_as_presheaf_categories}
    We describe multiple categories of graph as presheaf categories $\PresheafDefault$ and specify $\opcat{I}$ for each.
    \begin{itemize}[$\bullet$]
        \item 
        \emph{Multigraphs}: We allow parallel edges, i.e., with same source and target.
        \begin{itemize}[-]
            \item 
            Directed: denoted $\Graph$ in this paper;
            \(
                \opcat{I} =
                \smash{
                \begin{tikzcd}[ampersand replacement=\&]
                    E
                        \ar[r, "s"{description}, shift left=3pt]
                        \ar[r, "t"{description}, shift right=3pt]
                    \& V.
                \end{tikzcd}
                }
            \)

            \item
            Undirected \cite[§8]{vigna2003graphistopos}:
            each edge $e$ requires another one $\sym(e)$ in the other direction, which is equivalent to having an undirected edge.
            \[
                \opcat{I} =
                \smash{
                \begin{tikzcd}[ampersand replacement=\&]
                    E
                        \ar[loop, "\sym"', distance=1em, in=205, out=155]
                        \ar[r, "s"{description}, bend left = 0, shift left=3pt]
                        \ar[r, "t"{description}, bend right = 0, shift right=3pt]
                    \& V
                \end{tikzcd}
                }
                \quad
                \text{where}
                \qquad
                \begin{aligned}
                    \sym \cdot \sym &= \id{E}, \\
                    s \cdot \sym &= t, \\
                    t \cdot \sym &= s.
                \end{aligned}
            \]

            \item
            Directed reflexive \cite[§8]{vigna2003graphistopos}:
            Each vertex $v$ has a specific loop 
            \smash{\(
                \begin{tikzcd}
                    v 
                    \arrow[loop, "\refl_v"', distance=.8em, in=200, out=160]
                \end{tikzcd}
            \)}
            preserved by graph homomorphisms.
            It is equivalent to omitting those reflexive loops and considering that an edge mapped onto a reflexive loop is now sent to the vertex instead, giving \emph{degenerate} graphs.
            \[
                \opcat{I} =
                \smash{
                \begin{tikzcd}[ampersand replacement=\&]
                    E
                        \ar[r, "s"{description}, bend left = 0, shift left=6pt]
                        \ar[r, "t"{description}, bend right = 0, shift right=6pt]
                    \& V
                        \ar[l, "\refl"{description}]
                \end{tikzcd}
                }
                \quad
                \text{where}
                \qquad
                s \cdot \refl = t \cdot \refl = \id{V}.
            \]

            \item
            Undirected reflexive \cite[§8]{vigna2003graphistopos}: with both $\sym$ and $\refl$ in $\opcat{I}$, we must add the equation $\sym \cdot \refl = \refl$.
        \end{itemize}

        \item
        \emph{(Multi) Hypergraphs}:
        An \emph{hyperedge} can contain any number of vertices.
        \begin{itemize}[-]
            \item 
            For hyperedges to be ordered lists $\langle v_1, \ldots, v_{m} \rangle$ of vertices, let $\opcat{I}$ have objects $V$ and $E_m$ for each $m \in \N_{\geq 1}$.
            Homomorphisms preserve arities.
            \[
                \opcat{I}
                \text{ contains for each $m \in \N_{\geq 1}$:}
                \quad
                \smash{
                \begin{tikzcd}[ampersand replacement=\&, column sep=large]
                    E_m
                        \ar[r, "s_1"{description}, bend left = 0, shift left=6pt]
                        \ar[r, "\cdots"{description}, draw=none]
                        \ar[r, "s_{m}"{description}, bend right = 0, shift right=6pt]
                    \& V
                \end{tikzcd}
                }
            \]

            \item
            For hyperedges to be unordered sets $\set{v_1, \ldots, v_{m}}$, take the convention that all possible ordered list containing the same vertices must exist.
            For that, add $m$ symmetry arrows, i.e., $\opcat{I}$ contains for each $m \in \N_{\geq 1}$:
            \[
                \smash{
                \begin{tikzcd}[ampersand replacement=\&, column sep=large]
                    E_m
                        \ar[loop, "\sym_1"', distance=1em, in=115, out=65]
                        \ar[loop, "\cdots"{description}, distance=1em, in=205, out=155]
                        \ar[loop, "\sym_{m}"', distance=1em, in=295, out=245]
                        \ar[r, "s_1"{description}, bend left = 0, shift left=6pt]
                        \ar[r, "\cdots"{description}, draw=none]
                        \ar[r, "s_{m}"{description}, bend right = 0, shift right=6pt]
                    \& V
                \end{tikzcd}
                }
                \quad
                \text{where}
                \qquad
                \begin{aligned}
                    \sym_j\cdot \sym_j &= \id{E}, \\
                    s_1 \cdot \sym_1 &= s_2, \\
                    s_2 \cdot \sym_1 &= s_1, \\
                    s_i \cdot \sym_1 &= s_i \text{ for } i \neq 1,2, \\
                    \ldots
                \end{aligned}
            \]

            \item
            To consider \emph{$m$-uniform} hyperedges, i.e., having exactly $m$ vertices in all of them, for $m$ fixed, let $\opcat{I}$ have as objects only $E_m$ and $V$ in the previous examples.

            \item
            We can extend the previous hyperedge examples to admit targets \cite[Example 3.4]{Lowe_1993_SPO} by having for each $m,n$ an object $E_{m,n}$ with
            \[
                \opcat{I}(E_{m,n},V) 
                =
                \set{s_0,\ldots,s_{m-1}} 
                \cup
                \set{t_0, \ldots, t_{n-1}}.
            \]

            \item
            Here is an alternative definition of a hypergraph:
            Let $G$ be a hypergraph if it has three sets $G(R), G(E), G(V)$ and two functions $f : G(R) \to G(E)$ and $g : G(R) \to G(V)$.
            We read $r \in G(R)$ as meaning that $g(r)$ is a vertex incident to the hyperedge $f(r)$.
            Note that a vertex can be incident to same hyperedge multiple times and that morphisms do not preserve arities here. 
            \[
                \opcat{I} =
                \smash{
                \begin{tikzcd}[ampersand replacement=\&, sep=small]
                    E
                    \& R
                        \ar[l, "f"']
                        \ar[r, "g"]
                    \& V
                \end{tikzcd}
                }
            \]
        \end{itemize}
    \end{itemize}
\end{example}

\begin{example}
    As pointed out in \cite{Baez_Genovese_Master_Schulman_2021}, the category of Petri nets is not a topos because it is not cartesian closed.
    However, the authors show that the category of \emph{pre-nets}, which are Petri nets where the input and the output of a transition are ordered, is equivalent to a presheaf category.
\end{example}

\begin{example}
    \newcommand{\resources}{\mathbb{R}}
    \newcommand{\mathletter}{\text{\Letter}}
    \newcommand{\processes}{\mathbb{P}}
    \newcommand{\holds}{\mathit{holds}}
    \newcommand{\connections}{\mathbb{C}}
    \newcommand{\myconnect}{\mathit{connect}}
    \newcommand{\mykey}{%
        \begin{tikzpicture}[baseline=-0.5ex]
            \draw [-] (0,0) to ++(2.5mm,0mm) to ++(0,-.5mm);
            \draw [-] (2mm,0) to ++(0,-.5mm);
            \draw [fill=white] (0,0) circle (.6mm);
        \end{tikzpicture}%
    }
    Here is an example of how the fuzzy structure can be used in hypergraph rewriting.
    Assume we are in a distributed system setting, with $\resources \supseteq \{ \mathletter, \mykey \}$ a globally fixed set of resources. A system state $S = \langle \processes, \holds, \connections, \myconnect \rangle$ consists of: a set of processes $\processes$, a resource assignment $\holds : \processes \to \powerset(\resources)$, a set of ternary connections $\connections$, and a connection assignment $\myconnect : \connections \to \processes \times \processes \times \processes$.
    
    Suppose that in this system, a process $p$ can transmit $\mathletter$ to a process $q$ if, for some $C \in \connections$ and $r \in \processes$: (i)~$p$, $q$ and $r$ are distinct; (ii)~$\mathletter \in \holds(p)$; (iii)~$\mathletter \notin \holds(q)$; (iv)~$\mathletter \notin \holds(q)$; (v)~$\myconnect(C) = (p,q,r)$.
    Executing the transmission changes the state by removing $\mathletter$ from $\holds(p)$, adding $\mathletter$ to $\holds(q)$, removing $\mykey$ from $\holds(r)$ and removing $C$ from $\connections$.
    We can thus think of $r$ as a mediator for the transmission, which provides $\mykey$ as a required resource that is consumed.
    
    We can model states as objects of the fuzzy presheaf $\FuzzyPresheafDefault$ where:
    \(
     \opcat{I} =
        \smash{
        \begin{tikzcd}[cramped, column sep=20, ampersand replacement=\&]
        \connections 
            \arrow[r, "s" description, shift left=5.5pt] 
            \arrow[r, "m" description] 
            \arrow[r, "t" description, shift right=5.5pt]
        \& \processes
        \end{tikzcd}
        }
    \)
    and $\processes$ is labeled from the subset lattice $(\powerset(\resources), \subseteq)$, which is a complete Heyting algebra. Then the transmission described above can be informally defined as a hypergraph transformation rule:
    \begin{center}
    \begin{tikzpicture}[baseline=-3mm,->,node distance=12mm,nodes={rectangle}]
        \begin{scope}[]
              \node (p) {$p^{P \uplus \{\,\text{\raisebox{-.2ex}{\Letter }}\,\}}$};
              \node (q) [right=of p] {$q^{Q\, \setminus \{\,\text{\raisebox{-.2ex}{\Letter}}\,\}}$};
              \node (r) at ($(p)!.5!(q) + (0mm,-6mm)$) {$r^{R \uplus \{\mykey\}}$};
              \draw [->] (p) to (q);
              \draw [-,densely dotted] ($(p)!.5!(q)$) to (r);
        \end{scope}
        \node at (40mm,-3mm) {$\Longrightarrow$}; 
        \begin{scope}[xshift=50mm]
              \node (p) {$p^{P}$};
              \node (q) [right=of p] {$q^{Q \cup \{\,\text{\raisebox{-.2ex}{\Letter}}\,\}}$};
              \node (r) at ($(p.east)!.5!(q.west) + (0mm,-6mm)$) {$r^{R}$};
        \end{scope}
    \end{tikzpicture}
    \end{center}
    which is formally definable with a single \pbpostrong{} rule, so that it is applicable in any context. For extensive details on how this works for directed multigraphs, see \cite[Section 6]{Overbeek_Endrullis_Rosset_2023_PBPO+_Quasitopos}, which extends naturally to hypergraphs.
\end{example}

\begin{example}
    Other concrete examples of modellings using fuzzy multigraphs and fuzzy multigraph rewriting using \pbpostrong{} include:
    \begin{itemize}
        \item Binary decision diagrams reduction~\cite[Section 5]{Overbeek_2023_Tutorial_PBPOplus}.\footnote{To appear at the end of March 2023.}
        \item Processes that consume data from FIFO channels
        ~\cite[Example 80]{Overbeek_Endrullis_Rosset_2023_PBPO+_Quasitopos}.
        \item Terms and linear term rewriting~\cite[Section 4]{Overbeek_2021_Linear_term_rewriting_and_Termination}.
    \end{itemize}
\end{example}

\section{Fuzzy simple graphs form a quasitopos}
\label{sec:simple_fuzzy_graphs}

In \cref{ex:graph_categories_as_presheaf_categories}, we described different multigraph categories as presheaf categories.
Hence, adding a fuzzy structure to these categories results in quasitoposes by \cref{thm:fuzzy_presheaf_is_a_quasitopos}.
This raises the question whether simple fuzzy graph categories are quasitoposes, and if so, whether it can be shown using \cref{thm:fuzzy_presheaf_is_a_quasitopos}.

First, simple graphs do not arise as presheaf categories, meaning~\cref{thm:fuzzy_presheaf_is_a_quasitopos} cannot be applied.
Here is one way to observe this: presheaves are toposes and in toposes, all monomorphisms are regular.
However, monomorphisms in directed simple graphs \cite[p.\ 315]{Johnstone_Lack_Sobocinski_2007_Quasitoposes_quasiadhesive_artin_glueing} and in undirected simple graphs \cite[Lemma 3.7.1]{Plessas_2011_The_categories_of_graphs} are regular only if they reflect edges,\footnote{%
A graph homomorphism $f : G \to H$ \emph{reflects edges} if, for every $v,w \in G(V)$, every edge between $f_V(v)$ and $f_V(w)$ has an $f$-preimage.}
which is not always the case.

Nevertheless, we show in this section that directed and undirected simple fuzzy graph categories are quasitoposes, by using a different technique. In more detail,  given a \emph{(Lawvere-Tierney) topology} $\topology$ on a topos, the $\topology$-separated elements form a subcategory which is a quasitopos \cite[Theorem 10.1]{Johnstone_1979_On_a_topological_topos}.
Vigna \cite{vigna2003graphistopos} has used this fact to prove that the (directed/undirected) simple graphs are exactly the $\lnot \lnot$-separated elements of the respective (directed/undirected) multigraph category, where $\lnot \lnot$ is a topology. We extend the approach by Vigna, showing that directed and undirected simple fuzzy graphs are quasitoposes.

Because we start from fuzzy multigraphs, which form a quasitopos, we need to work with the more general definition of a \emph{topology on a quasitopos} (\cref{def:topology_on_quasitopos}).
The result that the separated elements form a quasitopos still holds true with this more general definition of a topology \cite[Theorem 43.6]{Wyler_1991}.
Recall \cref{not:pullback_along,def:subobject} if needed for the notations used.

\begin{remark}
    Alternatively, Johnstone et al.\ show that Artin Glueing~\cite{Carboni_Johnstone_1995_Artin_glueing} can be used to prove that \textit{directed} simple graphs form a quasitopos.
    Given a functor $T : \cat{C} \to \cat{D}$, denote by $\cat{C} \sslash T$ the full subcategory of the slice category $\cat{C} / T$ consisting of the monomorphisms $X \mono TY$.
    Then directed simple graphs can be seen as $\Set \sslash T$ for $TX = X \times X$.
    Because $\Set$ is a quasitopos and $T$ preserves pullbacks, we have by \cite[Theorem 16]{Carboni_Johnstone_1995_Artin_glueing} that $\Set \sslash T$, i.e., directed simple graphs, is a quasitopos.
\end{remark}


\begin{definition}[{\cite[Def.~41.1]{Wyler_1991}}]
    \label{def:topology_on_quasitopos}
    A \myemph{topology} $\topology$ on a quasitopos $\cat{E}$ is a family of mappings $\MonoClass(A) \to \MonoClass(A)$ for each $A \in \cat{E}$.
    It sends every monomorphism $m : A_0 \mono A$ to another monomorphism $\topology m : A_1 \mono A$ with same codomain $A \in \cat{E}$, such that
    \begin{enumerate}[label=(\roman*)]
        \item
        If $m \leq m'$, then $\topology m \leq \topology m'$.

        \item
        $m \leq \topology m$.

        \item
        $\topology \topology m \simeq \topology m$.

        \item
        For all $f:B \to A$, we have $\topology(\pullbackby{f}m) \simeq \pullbackby{f}(\topology m)$.

        \item
        If $m$ is regular, then so is $\topology m$.
    \end{enumerate}
\end{definition}

Notice that axioms $(i)-(iii)$ says that $\topology$ is a closure operation.
For instance in $\Set$, given a subset $A_0 \subseteq A$, $(ii)$ requires that $A_0 \subseteq \topology(A_0) \subseteq A$.
By $(iv)$ and $(v)$ the closure operation must commute with pullbacks and preserve regularity.

\begin{example}[{\cite[Ex.~41.2]{Wyler_1991}}]
    We give two basic examples of topologies.
    The third example is more advanced; examples of it in a few categories are given in the next lemma.
    \begin{itemize}[topsep=2pt]
        \item
        The \emph{trivial topology}
        $\topology(A_0 \stackrel{m}\mono A) \defeq (A \stackrel{\id{A}}{\longequal} A)$.
        
        \item 
        The \emph{discrete topology}
        $\topology(A_0 \stackrel{m}\mono A) \defeq (A_0 \stackrel{m}\mono A)$.

        \item
        Let $m : A_0 \mono A$.
        Consider the (necessarily monic) morphism from the initial object $0 \mono A$ and factorise it as (epi mono, regular mono) $\overline{o_A} \circ e : A_0 \epimono A_1 \regmono A$.
        Let $\lnot m$ be the exponential object $(\overline{o_A} : A_1 \regmono A)^{(m:A_0 \mono A)}$ in the slice category $\cat{E}/A$.
        Then $\lnot \lnot$ is a topology called the \emph{double negation}.
    \end{itemize}
\end{example}

\begin{lemma}
    \label{lem:details_notnot_topology}
    We detail the $\lnot \lnot$ topology in a few categories.
    In each situation, we consider subobjects of an object $A$ or $(A,\alpha)$.
    \begin{itemize}[noitemsep]
        \item 
        In $\Set$, 
        \(
            \lnot A_0 = A \setminus A_0.
        \)
        Therefore, $\lnot \lnot = \id{}$.

        \item
        In $\FuzzySetDefault$, 
        \(
            \lnot(A_0,\alpha_0) = (A \setminus A_0, \alpha).
        \)
        Therefore, $\lnot \lnot$ is the identity on the subset but replaces its membership function $\alpha_0$ with $\alpha$, the one of $A$.
        
        \item
        In $\Graph$,
        $\lnot A_0$ is the largest subgraph of $A$ totally disconnected from $A_0$.
        Hence, $\lnot \lnot A_0$ is the largest subgraph of $A$ induced by $A_0(V)$, i.e., every edge from $A$ with source and target in $A_0$ is added in $\lnot \lnot A_0$. 

        \item
        In $\FuzzyGraph$, $\lnot \lnot$ acts like in $\Graph$ on the underlying graph, i.e., it adds all edges that had source and target in the subgraph, and it replaces the membership function of the subobject with the one of the main object.
    \end{itemize}
\end{lemma}

The notions of density and separatedness from standard topology, can also be defined for (quasi)topos topologies.
To illustrate what density is, recall that $\Q$ is called dense in $\R$ because closing it under the standard metric topology, i.e., adding an arbitrary small open interval around each rational number, gives the whole space $\R$.
In standard topology, recall also that every separated space $B$ (a.k.a.\ Hausdorff space), has the property that functions $f : A \to B$ are fully determined by the images on any dense subsets of $A$.

\begin{definition}[{\cite[Def.~41.4, 42.1]{Wyler_1991}}]
    \label{def:dense_subobjects_and_separated_elements}
    Given a topology $\topology$ in a quasitopos $\cat{E}$:
    \begin{itemize}[topsep=2pt]
        \item 
        a monomorphism $A_0 \stackrel{m}{\mono} A$ is \myemph{$\topology$-dense} if $\topology(A_0 \stackrel{m}{\mono} A) = (A \stackrel{\id{A}}{\longequal} A)$,
    \end{itemize}
        
    \begin{mysidepicture}{2.4cm}{-1ex}{-.5ex}{
                \begin{tikzcd}[ampersand replacement=\&, sep=small]
                    A_0 \& \\
                    A \& B
                    \ar[from=1-1, to=2-2, "f"]
                    \ar[from=1-1, to=2-1, tail, "m"']
                    \ar[from=2-1, to=2-2, dotted, "g"']
                \end{tikzcd}
    }
    \begin{itemize}[topsep=0pt]    
        \item
        an object $B$ is called \myemph{$\topology$-separated} if for every $\topology$-dense subobject $m : A_0 \mono A$ and every morphism $f : A_0 \to B$ there exists at most one factorisation $g : A \to B$ of $f$ through $m$.
    \end{itemize}
    \end{mysidepicture}
\end{definition}

The $\lnot \lnot$-separated graphs are precisely the simple graphs.

\begin{lemma}[{\cite[Thms.~3 \& 4]{vigna2003graphistopos}}]
    \label{lem:separated_elements_in_multigraphs_are_simple_graphs}
    In $\Graph$, a subgraph $A_0 \subseteq A$ is $\lnot \lnot$-dense if it contains all the vertices $A_0(V) = A(V)$.
    As a consequence, a graph $B$ is $\lnot \lnot$-separated if it has no parallel edges.
\end{lemma}

Similarly, the $\lnot \lnot$-separated fuzzy graphs are the simple fuzzy graphs.

\begin{lemma}
    \label{lem:separated_elements_in_fuzzy_multigraphs_are_fuzzy_simple_graphs}
    In the category of (directed and undirected) fuzzy graphs, a subgraph $(A_0,\alpha_0) \subseteq (A,\alpha)$ is $\lnot \lnot$-dense if it contains all the vertices $A_0(V) = A(V)$.
    Hence, a fuzzy graph $(B,\beta)$ is $\lnot \lnot$-separated if it has no parallel edges.
    \customqed
\end{lemma}

\begin{corollary}
    Directed and undirected simple fuzzy graphs form quasitoposes.
\end{corollary}

\begin{proof}
    Immediate by \cref{lem:separated_elements_in_fuzzy_multigraphs_are_fuzzy_simple_graphs} and \cite[Theorem 43.6]{Wyler_1991}.
\end{proof}

\section{Conclusion}
\label{sec:conclusion}

In this paper, we introduced the concept of fuzzy presheaves and proved that they form rm-adhesive quasitoposes.
Furthermore, we showed that simple fuzzy graphs, both directed and undirected, also form quasitoposes.

There are several directions for future work.
The alternative definition of fuzzy presheaves from \cref{rem:alternate_definition_of_fuzzy_presheaf}, with poset morphisms, can be explored further.
For graphs, it means that the membership of an edge gives a lower-bound for or determines the membership of its source and target.

To obtain fuzzy presheaves, we have added a pointwise fuzzy structure to presheaves.
More generally, different notions of fuzzy categories have been defined \cite{Sugeno_1983_Fuzzy_sets_and_systems,Syropoulos_2020_A_modern_introduction_to_fuzzy_mathematics}.
It is natural to wonder if a fuzzy (quasi)topos is also a quasitopos.

Finally, a more abstract question is whether having a fuzzy structure is a particular instance of a more abstract categorical construction.
For instance, one could express an inequality in a poset $\mathcal{L}$ via a natural transformations and then say that a diagram is $\leq$-commuting when it commutes up to this natural transformation.
Another possibility would be to consider poset-enriched categories for comparing morphisms via inequality.
Fuzzy structures do look like slice or comma categories but do not precisely fall under their scope.
For comma categories, Artin Glueing gives a nice criteria for obtaining new quasitoposes~\cite{Johnstone_Lack_Sobocinski_2007_Quasitoposes_quasiadhesive_artin_glueing}.
We wonder if more abstract uniform theorems for obtaining quasitoposes can be expressed and proved.

\subsubsection*{Acknowledgments}
We thank Helle Hvid Hansen for discussions and valuable suggestions.
The authors received funding from the Netherlands Organization for Scientific Research (NWO) under the Innovational Research Incentives Scheme Vidi (project.\ No.\ VI.Vidi.192.004).

\newpage

\bibliographystyle{includes/splncs04}
\bibliography{main}

\newpage

\section{Appendix}


\subsection{Proof of \cref{lem:chi_A_1)is_natural_2)makes_a_pullback_square_3)is_unique}}

\begin{lemma}
    Given $j \xrightarrow{\iota} i$ and a sieve $A \stackrel{n}{\rightarrowtail} y(i)$,
    \begin{equation}
        \label{eq:formula_Omega(iota)(n)=y(j)_iff   _iota_in_A(j)}
        \Omega(\iota)(n) = y(j) \iff j \xrightarrow{\iota} i \in A(j).
    \end{equation}
\end{lemma}

\begin{proof}
    Observe first the following:
    \begin{equation}
        \label{eq:formula_iota_in_A(j)_iff_forall_k_iotakappa_in_A(k)}
        j \xrightarrow{\iota} i \in A(j) \iff \forall k \in I, \forall k \xrightarrow{\kappa} j : \ \  k \xrightarrow{\kappa} j \xrightarrow{\iota} i \in A(k). 
    \end{equation}
    The direction $(\Rightarrow)$ follows from the definition of $A$: as a subfunctor of $y(i)$ its effect on morphisms is precomposition.
    Hence, $A(\kappa)(\iota) = \iota \circ \kappa$ needs to be in $A(k)$ for well-definedness.
    Conversely, $(\Leftarrow)$ follows from taking $k=j$ and $\kappa=\id{j}$.
    
    Name $A_0$ the pullback of $n$ along $y(\iota)=\iota \circ -$ and take $k \in I$.
    Observe that the pullback is
    \begin{align*}
        A_0(k)
        &= \setvbar{\big(k \xrightarrow{\kappa} j, a \in A(k)\big)} {k \xrightarrow{\kappa} j \xrightarrow{\iota} i = n_k(a)} \\
        &\isom \setvbar{k \xrightarrow{\kappa} j} {k \xrightarrow{\kappa} j \xrightarrow{\iota} i \in \ima(n_k) \isom A(k)}
    \end{align*}
    Therefore,
    \begin{align*}
        \Omega(\iota)(n) = A_0 = y(j)
        &\iff \forall k \in I : \ \ A_0(k) = I(k,j) \\
        &\iff \forall k \in I, \forall k \xrightarrow{\kappa} j : \ \ k \xrightarrow{\kappa} j \xrightarrow{\iota} i \in A(k) \\
        &\stackrel{(\ref{eq:formula_iota_in_A(j)_iff_forall_k_iotakappa_in_A(k)})}{\iff}
        j \xrightarrow{\iota} i \in A(j).   
    \end{align*}
\end{proof}

\begin{proof}[Proof of \cref{lem:chi_A_1)is_natural_2)makes_a_pullback_square_3)is_unique}]
    \begin{enumerate}
        \item 
        Take $\iota : j \to i$ in $I$.
        We show the following commute
        \[\begin{tikzcd}[ampersand replacement=\&]
        	{B(i)} \& {\Sub_{\PresheafHatDefault}(I(-,i))} \\
        	{B(j)} \& {\Sub_{\PresheafHatDefault}(I(-,j))}
        	\arrow["{B(\iota)}"', from=1-1, to=2-1]
        	\arrow["{\Sub_{\PresheafHatDefault}(y(\iota))}", from=1-2, to=2-2]
        	\arrow["{(\chi_A)_i}", from=1-1, to=1-2]
        	\arrow["{(\chi_A)_j}"', from=2-1, to=2-2]
        \end{tikzcd}\]
        Take $b \in B(i)$.
        We are comparing presheaves $\opcat{I} \to \Set$, hence take also $k \in I$.
        \begin{align*}
            \Big( \Sub_{\PresheafHatDefault}(y(\iota)) \circ (\chi_A)_i \Big) (b)(k)
            &\stackrel{\text{def.~}\chi_A}{=}
            \Sub_{\PresheafHatDefault}(y(\iota)) \big(\setvbar{\kappa:k \to i}{B(\kappa)(b) \in m_k(A_k)}\big) \\
            &\stackrel{\text{def.~}\Sub_{\PresheafHatDefault}}{=}
            \{(k \xrightarrow{\rho} j, k \xrightarrow{\kappa} i) ~\vert~ B(\kappa)(b) \in m_k(A_k) \text{ and } \\
            & \hspace*{11.6em}  k \xrightarrow{\rho}j \xrightarrow{\iota} i = k \xrightarrow{\kappa}i \} \\
            &\stackrel{\hphantom{\text{def.~}\chi_A}}{\isom} 
            \setvbar{\rho : k \to j}{B(\rho)\big(B(\iota)(b)\big) \in m_k(A_k)} \\
            &\stackrel{\text{def.~}\chi_A}{=}
            \Big( (\chi_A)_j \circ B(\iota) \Big) (b)(k) 
        \end{align*}

        \item
        The pullback of $\chi_A$ and $\true$ is, for $i \in I$:
        \begin{align*}
            &\setvbar{\big(b \in B(i), \cdot \in 1 \big)}
            {(\chi_A)_i(b) = \true_i(\cdot)} \\
            &\isom
            \setvbar{b \in B(i)}
            {\forall j \in I: \ (\chi_A)_i(b)(j) = \true_i(\cdot)(j) = y(i)(j) = I(j,i)} \\
            &=
            \setvbar{b \in B(i)}
            {\forall j \in I: \ \setvbar{j \xrightarrow{\iota} i}{B(\iota)(b) \in A(j)} = I(j,i)} \\
            &=
            \setvbar{b \in B(i)}
            {\forall j \in I, \forall j \xrightarrow{\iota} i: \ B(\iota)(b) \in A(j)} \\
            &= A(j).
        \end{align*}
        The last equality hold since: $(\supseteq)$ follows from $B(\iota)(b) = A(\iota)(b)$ when $b \in A(i)$ since $A$ is a subobject of $B$, and $(\subseteq)$ by taking $\iota = \id{i}$.

        \item
        Suppose another natural transformation $\nu : B \to \Omega$ such that the pullback of $\nu$ and $\true$ gives $A \subseteq B$ and $! : A \to 1$.
        Consider the element $B(\iota)(b)$ in $B(j)$ and the condition for it to be in the subobject $A(j)$:
        \begin{align*}
            B(\iota)(b) \in A(j) 
            &\iff
            \nu_j\big(B(\iota)(b)\big) = \true_j(\cdot)
            \tag{pullback} \\
            &\iff 
            \Omega(\iota)(\nu_i(b)) = y(j)
            \tag{$\nu$ nat.~and $\true$ def.} \\
            &\iff \iota \in \nu_i(b)(j).
            \tag{by (\ref{eq:formula_Omega(iota)(n)=y(j)_iff   _iota_in_A(j)})}
        \end{align*}
        If we take back the definition of $\chi_A$, we have thus
        \begin{align*}
            (\chi_A)_i(b)(j)
            &= \setvbar{j \xrightarrow{\iota} i}{B(\iota)(b) \in A_j}
            \tag{def.~\ref{eq:def_chi_A}} \\
            &\isom \setvbar{j \xrightarrow{\iota} i}{\iota \in \nu_i(b)(j)}
            \tag{equivalence above} \\
            &\isom \nu_i(b)(j).
        \end{align*}
        Hence, $\chi_A$ and $\nu$ are equal, proving the uniqueness.
    \end{enumerate}
\end{proof}

\subsection{Proof of \cref{lem:fuzzy_presheaf_exponential_adjunction}}

\noindent We exhibit a natural isomorphism to show the exponentiation adjunction.
\begin{center}
    \begin{tikzcd}
        \Hom\Big((C,\gamma) \times (A,\alpha), (B,\beta)\Big)
            \ar[r, shift left=.4em, "\phi", "\isom"'] 
        & \Hom\Big((C,\gamma), (B,\beta)^{(A,\alpha)} \Big)
            \ar[l, shift left=.4em, "\psi"]
    \end{tikzcd}
\end{center}

\begin{definition}
    \label{def:phi}
    Let $\phi$ be defined as follows.
    \[
        \biggl(
        \begin{tikzcd}[ampersand replacement=\&, row sep = 0.15em, column sep = tiny]
        	{C \times A} \&\& B \\
        	\& {\color{gray}\labels}
        	\arrow["h", from=1-1, to=1-3]
        	\arrow["{\gamma \meet \alpha}"', color=gray, from=1-1, to=2-2]
        	\arrow["\beta", color=gray, from=1-3, to=2-2]
        \end{tikzcd}
        \biggr)
        \quad
        \xmapsto{\phi}
        \quad
        \biggl(
        \begin{tikzcd}[ampersand replacement=\&, row sep = 0.15em, column sep = tiny]
        	C \&\& \PresheafHatDefault(y(-) \times A, B) \\
        	\& {\color{gray}\labels}
        	\arrow["\phi(h)", from=1-1, to=1-3]
        	\arrow["\gamma"', color=gray, from=1-1, to=2-2]
        	\arrow["\theta", color=gray, from=1-3, to=2-2]
        \end{tikzcd}
        \biggr)
    \]
    For $i,j \in I$ and $c \in C(i)$, let $\phi(h)_i(c)_j: y(i)(j) \times A(j) \rightarrow B(j)$ be
    \begin{equation}
        (j \xrightarrow{\iota} i, a) \mapsto h_j( C(\iota)(c), a). \label{eq:def_phi}
    \end{equation}
\end{definition}

\begin{lemma}
    \label{lem:phi_well_defined}
    The natural transformation $\phi$ from \cref{def:phi} is well-defined:
    \begin{enumerate}
        \item
        $\phi$ is natural in $C$ and $B$,

        \item
        $\phi(h)_i(c)$ is natural in $j \in I$, for any $i \in I$ and $c \in C(i)$, and

        \item
        $\gamma \leq \theta \cdot \phi(h)$.
    \end{enumerate}
\end{lemma}

\begin{proof}
    \begin{enumerate}
        \item 
        We verify that $\phi$ is natural in $C$ and $B$.
        Take $f: (C',\gamma') \to (C,\gamma)$ and $g: (B,\beta) \to (B', \beta')$.
        
        \noindent\begin{tabularx}{\textwidth}{@{}XX@{}}
            \begin{equation}
                \label{eq:f}
                \begin{tikzcd}[scale cd =.8, ampersand replacement=\&]
                	{C'} \& {} \& C \\
                	\& {\color{gray}\labels}
                	\arrow["{\gamma'}"', color=gray, from=1-1, to=2-2]
                	\arrow["f", from=1-1, to=1-3]
                	\arrow["\gamma", color=gray, from=1-3, to=2-2]
                        \arrow["\leq" description, draw=none, color=gray, from=1-2, to=2-2]
                \end{tikzcd}
            \end{equation} &
            \begin{equation}
                \label{eq:g}
                \begin{tikzcd}[scale cd =.8, ampersand replacement=\&]
                	B \& {} \& {B'} \\
                	\& {\color{gray}\labels}
                	\arrow["\beta"', color=gray, from=1-1, to=2-2]
                	\arrow["g", from=1-1, to=1-3]
                	\arrow["{\beta'}", color=gray, from=1-3, to=2-2]
                        \arrow["\leq" description, draw=none, color=gray, from=1-2, to=2-2]
                \end{tikzcd}
            \end{equation}
        \end{tabularx}
        We want to prove that the following commute:
        \[\begin{tikzcd}[column sep = huge]
        	{\Hom \big( (C,\gamma) \times (A,\alpha),(B,\beta) \big)} & {\Hom \big( (C,\gamma), (B,\beta)^{(A,\alpha)} \big)} \\
        	{\Hom \big( (C',\gamma') \times (A,\alpha),(B',\beta') \big)} & {\Hom \big( (C',\gamma'), (B',\beta')^{(A,\alpha)} \big)}
        	\arrow["{g \circ - \circ (f \times \id{A})}"', from=1-1, to=2-1]
        	\arrow["{g^{(A,\alpha)} \circ - \circ f}", from=1-2, to=2-2]
        	\arrow["{\phi_{(C,\gamma), (B,\beta)}}", from=1-1, to=1-2]
        	\arrow["{\phi_{(C',\gamma'), (B',\beta')}}"', from=2-1, to=2-2]
        \end{tikzcd}\]
        We check it.
        Take $i \in I, c' \in C'(i), j \in I, j \xrightarrow{\iota} i$ and $a \in A(j)$:
        \begin{align*}
            \phi \big( gh(f\times\id{A})_i (c') \big)_j (\iota, a)
            &= \big( gh(f\times\id{A}) \big)_j (C'(\iota)(c'), a)
            \tag{by \eqref{eq:def_phi}} \\
            &= g_j h_j \big( f_j \cdot C'(\iota) (c'), a \big)
            \\
            &= g_j h_j \big( C(\iota)(f_i(c')), a \big)
            \tag{$f$ nat.} \\
            &= g_j \Big( \phi(h)_i \big(f_i(c') \big) \Big)_j (\iota, a)
            \tag{by \eqref{eq:def_phi}} \\
            &= \Big(  g \cdot \phi(h)_i \big( f_i(c') \big) \Big)_j (\iota, a)
            \\
            &= g^{(A,\alpha)}_i \Big( \phi(h)_i \big( f_i(c') \big) \Big)_j (\iota, a)
            \tag{by \eqref{eq:exponent_def_morphisms}} \\
            &= \big( (g^{(A,\alpha)} \cdot \phi(h) \cdot f)_i (c') \big)_j (\iota, a)
        \end{align*}
    
        \item
        We verify that $\phi(h)_i(c)$ is natural in $j$: given $k \xrightarrow{\kappa} j$
        \[\begin{tikzcd}[scale cd=.75, row sep = .5em]
        	{y(i)(j) \times A(j)} &&&& {B(j)} \\
        	& {(j \xrightarrow{\iota}i,a)} && {h_j\big(C(\iota)(c),a\big)} \\
        	\\
        	\\
        	&&& {B(\kappa)\Big( h_j \big(C(\iota)(c),a\big)  \Big)} \\
        	& {\left(k \xrightarrow{\kappa} j \xrightarrow{\iota}i, A(\kappa)(a)\right)} && {h_k \big( C(\iota\kappa)(c), A(\kappa)(a) \big)} \\
        	{y(i)(k) \times A(k)} &&&& {B(k)}
        	\arrow["{y(i)(\kappa) \times A(\kappa)}"', from=1-1, to=7-1]
        	\arrow["{B(\kappa)}", from=1-5, to=7-5]
        	\arrow["{\big(\phi(h)_i(c)\big)_k}"', from=7-1, to=7-5]
        	\arrow["{\big(\phi(h)_i(c)\big)_j}", from=1-1, to=1-5]
        	\arrow[maps to, from=2-2, to=6-2]
        	\arrow[maps to, from=2-2, to=2-4]
        	\arrow[maps to, from=2-4, to=5-4]
        	\arrow[maps to, from=6-2, to=6-4]
        	\arrow["\text{nat.~} h",equal, from=5-4, to=6-4]
        \end{tikzcd}\]
        
        \item
        Since $h:(C\times A, \gamma \meet \alpha) \to (B,\beta)$, we have $\gamma \meet \alpha \leq \beta h$.
        \begin{align*}
            &\Rightarrow
            \forall i \in I, \forall c \in C(i), \forall a \in A(i): \gamma_i(c) \meet \alpha_i(a) \leq \beta_i h_i (c,a)
            \\
            &\stackrel{\smash{\eqref{eq:heyting_algebra_cartesian_closed}}}{\Rightarrow}
            \forall i \in I, \forall c \in C(i), \forall a \in A(i): \gamma_i(c) \leq \big( \alpha_i(a) \Rightarrow \beta_i h_i (c,a) \big)
            \\
            &\stackrel{\smash{\eqref{eq:def_phi}}}{\Rightarrow}
            \forall i \in I, \forall c \in C(i), \forall a \in A(i): \gamma_i(c) \leq \big( \alpha_i(a) \Rightarrow \beta_i (\phi(h)_i(c))_i (\id{i},a) \big) 
            \\
            &\Rightarrow
            \forall i \in I, \forall c \in C(i): \gamma_i(c) \leq \textstyle{\bigmeet}_{a \in A(i)} \big( \alpha_i(a) \Rightarrow \beta_i (\phi(h)_i(c))_i (\id{i},a) \big)
            \\
            &\stackrel{\smash{\eqref{eq:exponent_def_morphisms}}}{\Rightarrow}
            \forall i \in I, \forall c \in C(i): \gamma_i(c) \leq \theta_i \big( \phi(h)_i(c) \big)
            \\
            &\Rightarrow
            \gamma \leq \theta \cdot \phi(h)
            \qedhere
        \end{align*}
    \end{enumerate}
\end{proof}

\begin{definition}
    \label{def:psi}
    Let $\psi$ be defined on $i \in I$ as follows.
    \[
        \biggl(
        \begin{tikzcd}[ampersand replacement=\&, row sep = 0.15em, column sep = tiny]
        	C \&\& \PresheafHatDefault(y(-) \times A, B) \\
        	\& {\color{gray}\labels}
        	\arrow["k", from=1-1, to=1-3]
        	\arrow["\gamma"', color=gray, from=1-1, to=2-2]
        	\arrow["\theta", color=gray, from=1-3, to=2-2]
        \end{tikzcd}
        \biggr)
        \quad
        \xmapsto{\psi}
        \quad
        \biggl(
        \begin{tikzcd}[ampersand replacement=\&, row sep = 0.15em, column sep = tiny]
        	{C \times A} \&\& B \\
        	\& {\color{gray}\labels}
        	\arrow["\psi(k)", from=1-1, to=1-3]
        	\arrow["{\gamma \meet \alpha}"', color=gray, from=1-1, to=2-2]
        	\arrow["\beta", color=gray, from=1-3, to=2-2]
        \end{tikzcd}
        \biggr)
    \]
    For $i \in I, c \in C(i)$ and $a \in A(i)$, let
    \begin{equation}
        \psi(k)_i(c, a) \defeq k_i(c)_i (\id{i},a).
        \label{eq:def_psi}
    \end{equation}
\end{definition}

\begin{lemma}
    \label{lem:psi_well_defined}
    The natural transformation $\psi$ from \cref{def:psi} is well-defined:
    \begin{enumerate}
        \item
        $\psi$ is natural in $C$ and $B$
        
        \item
        \label{item:psi_well_defined_1}
        $\psi(k)$ is natural in $i \in I$, and

        \item
        \label{item:psi_well_defined_2}
        $\gamma \meet \alpha \leq \beta \cdot \psi(k)$.
    \end{enumerate}
\end{lemma}

\begin{proof}
    \begin{enumerate}
        \item
        This follows from $\phi$ being natural in $C$ and $A$ and $\psi$ being the inverse of $\phi$ as proved below in \cref{lem:phi_psi_are_inverses_of_each_other}.
        
        \item 
        We verify that $\psi(k)$ is natural in $i$: given $j \xrightarrow{\iota} j$
        \[\begin{tikzcd}[scale cd=.75, row sep = .5em]
        	{C(i) \times A(i)} &&&& {B(i)} \\
        	& {(c,a)} && {k_i(c)_i(\id{i},a)} \\
        	\\
        	\\
        	&&& {B(\iota) \big( k_i(c)_i (\id{i},a) \big)} \\
        	& {\big( C(\iota)(c) , A(\iota)(a) \big)} && {k_j\big( C(\iota)(c) \big)_j \big(\id{j}, A(\iota)(a) \big)} \\
        	{C(j) \times A(j)} &&&& {B(j)}
        	\arrow["{C(\iota) \times A(\iota)}"', from=1-1, to=7-1]
        	\arrow["{B(\iota)}", from=1-5, to=7-5]
        	\arrow["{\psi(k)_i}"', from=7-1, to=7-5]
        	\arrow["{\psi(k)_j}", from=1-1, to=1-5]
        	\arrow[maps to, from=2-2, to=6-2]
        	\arrow[maps to, from=2-2, to=2-4]
        	\arrow[maps to, from=2-4, to=5-4]
        	\arrow[maps to, from=6-2, to=6-4]
        	\arrow["(\ast)",equal, from=5-4, to=6-4]
        \end{tikzcd}\]
        where $(\ast)$ is the following reasoning:
        \begin{align*}
            B(\iota)\big( k_i(c)_i (\id{i},a) \big) 
            &= k_i(c)_j \cdot \big( y(i)(\iota) \times A(\iota) \big) (\id{i},a)
            \tag{$k_i(c)$ nat.} \\
            &= k_i(c)_j \big( \iota, A(\iota)(a) \big)
            \tag{Def.~\ref{def:contravariant_hom_functor}} \\
            &= \big( k_i(c) \cdot (y(\iota) \times \id{A}) \big)_j (\id{j}, A(\iota)(a))
            \tag{Def.~\ref{def:yoneda_embedding}} \\
            &= \big( \PresheafHatDefault(y(\iota) \times A,B) \cdot k_i(c) \big)_j (\id{j}, A(\iota)(a)) 
            \tag{by \eqref{eq:exponent_def_objects_iota}} \\
            &= k_j \big( C(\iota)(c) \big)_j \big(\id{j}, A(\iota)(a) \big)
            \tag{$k$ nat.}
        \end{align*}

        \item
        Since $k:(C,\gamma) \to ( \PresheafHatDefault(y(-) \times A,B) , \theta )$, we have $\gamma \leq \theta k$.
        \begin{align*}
            &\Rightarrow
            \forall i \in I, \forall c \in C(i): \gamma_i(c) \leq \theta_i (k_i(c))
            \\
            &\stackrel{\smash{\eqref{eq:exponent_def_objects_i}}}{\Rightarrow}
            \forall i \in I, \forall c \in C(i): \gamma_i(c) \leq \textstyle{\bigmeet}_{a \in A(i)} \left( \alpha_i(a) \Rightarrow \beta_i \cdot k_i(c)_i (\id{i},a) \right)
            \\
            &\Rightarrow
            \forall i \in I, \forall c \in C(i), \forall a \in A(i): \gamma_i (c) \leq \big( \alpha_i (a) \Rightarrow \beta_i \cdot k_i(c)_i (\id{i},a) \big)
            \\
            &\stackrel{\smash{\eqref{eq:heyting_algebra_cartesian_closed}}}{\Rightarrow}
            \forall i \in I, \forall c \in C(i), \forall a \in A(i): \gamma_i (c) \meet \alpha_i(a) \leq \beta_i \cdot k_i(c)_i (\id{i},a)
            \\
            &\stackrel{\smash{\eqref{eq:def_psi}}}{\Rightarrow}
            \forall i \in I, \forall c \in C(i), \forall a \in A(i): \gamma_i (c) \meet \alpha_i(a) \leq \beta_i \cdot \psi(k)_i (c,a)
            \\
            &\Rightarrow
            \gamma \meet \alpha \leq \beta \cdot \psi(k)
            \qedhere
        \end{align*}
    \end{enumerate}
\end{proof}

\begin{lemma}
    \label{lem:phi_psi_are_inverses_of_each_other}
    Both $\phi$ and $\psi$ are inverses of each other.
\end{lemma}

\begin{proof}
    We first verify that $\psi \phi = \id{}$.
    Take $(C \times A, \gamma \meet \alpha) \xrightarrow{h} (B,\beta)$.
    For $i \in I, c \in C(i)$ and $a \in A(i)$, we have
    \begin{align*}
        \psi(\phi(h))_i (c,a)
        &= \big( \phi(h)_i(c) \big)_i (\id{i},a)
        \tag{by \eqref{eq:def_psi}} \\
        &= h_i \big( C(\id{i})(c) , a \big)
        \tag{by \eqref{eq:def_phi}} \\
        &= h_i (c,a)
    \end{align*}
    
    We now verify that $\phi \psi = \id{}$.
    Take $(C,\gamma) \xrightarrow{k} (\PresheafHatDefault(y(-) \times A, B) , \theta)$.
    For $i \in I, c \in C(i), j \in I, j \xrightarrow{\iota} i$ and $a \in A(j)$, we have
    \begin{align*}
        \big( \phi(\psi(k))_i (c) \big)_j (\iota,a)
        &= \psi(k)_j \big( C(\iota)(c) , a \big)
        \tag{by \eqref{eq:def_phi}} \\
        &= k_j \big( C(\iota)(c) \big)_j (\id{j}, a)
        \tag{by \eqref{eq:def_psi}} \\
        &= \big( k_j \cdot C(\iota) (c) \big)_j (\id{j}, a)
        \\
        &= \big( \PresheafHatDefault( y(\iota) \times A , B ) \cdot k_i (c) \big)_j (\id{j}, a)
        \tag{$k$ nat.} \\
        &= \big( k_i(c) \cdot (y(\iota) \times \id{A}) \big)_j (\id{j}, a)
        \tag{by \eqref{eq:exponent_def_objects_iota}} \\
        &= k_i (c)_j \cdot (y(\iota) \times \id{A})_j (\id{j},a)
        \\
        &= k_i (c)_j (\iota, a)
        \tag{Def.~\ref{def:contravariant_hom_functor}}
    \end{align*}
\end{proof}

\subsection{Proof of \cref{lem:equivalence_of_categories}}

To prove \cref{lem:equivalence_of_categories}, we construct a functor in each direction and then prove that they are quasi-inverses.

\begin{definition}
    [Functor F]
    \label{def:F}
    We define the mapping
    \[
        F:\quot{\FuzzyPresheafDefault}{(D,\delta)}
        \to
        \FuzzyPresheaf{\elements{D}}{\tilde{\labels}}
    \]
    by $F(A,\alpha,p) \defeq (\tilde{A},\tilde{\alpha})$, where $\tilde{A}(i,d) \defeq p_i\inv(d)$,
    \begin{align}
        \tilde{A} \big(\iota: (j,e) \to (i,d) \big)
        \defeq
        A(\iota)\restrict{p_i\inv(d)} : p_i\inv(d) &\to p_j\inv(e), 
        \nonumber \\
        \tilde{\alpha}_{(i,d)} 
        \defeq 
        \alpha_i\restrict{p_i\inv(d)} : p_i \inv(d) &\to \labels(i)_{\leq \delta_i(d)},
        \label{eq:def_tilde_alpha}
    \intertext{and $F(f:(A,\alpha,p) \to (B,\beta,q)) \defeq \tilde{f} : (\tilde{A},\tilde{\alpha}) \to (\tilde{B},\tilde{\beta})$, where}
        \tilde{f}_{(i,d)} \defeq f_i\restrict{p_i\inv(d)} : p_i \inv(d) &\to q_i \inv(d).
        \label{eq:def_tilde_f}
    \end{align}




\end{definition}

\begin{lemma}
    \label{lem:F_is_well_defined}
    $F$ is a well-defined functor.
\end{lemma}

\begin{proof}
    Take $i \in I, d \in D(i)$ and $a \in p_i \inv (d)$.
    For $\tilde{f}_{(i,d)}$ to be well-defined, we first verify that
    \begin{equation} \label{eq:f_i_sends_preimage_on_preimage}
        f_i(p_i \inv(d)) \subseteq q_i \inv (d).
    \end{equation}
    This holds since
    \begin{align*}
        a \in p_i \inv (d) 
        &\iff
        p_i (a) = d \\
        &\iff
        q_i f_i (a) = d 
        \tag{$p=qf$}\\
        &\iff
        f_i(a) \in q_i \inv (d).
    \end{align*}
    Second, the naturality of $\tilde{f}$ result from the naturality of $f$, since the restrictions do not change anything.
    We also check that $\tilde{\alpha} \leq \tilde{\beta} \tilde{f}$:
    \begin{align*}
        \tilde{\alpha}_{(i,d)} (a)
        &= \alpha_i \restrict{\scaleto{p_i^{-1}  (d)}{8pt}} (a)
        \tag{def.~$\tilde{\alpha}$ \eqref{eq:def_tilde_alpha}} \\
        &\leq (\beta f)_i \restrict{\scaleto{p_i \inv (d)}{8pt}} (a)
        \tag{def.~$f$} \\
        &= \beta_i \restrict{\scaleto{q_i \inv (d)}{8pt}} f_i \restrict{\scaleto{p_i \inv (d)}{8pt}} (a)
        \tag{by \eqref{eq:f_i_sends_preimage_on_preimage}} \\
        &= \tilde{\beta}_{(i,d)} \tilde{f}_{(i,d)} (a).
        \tag{defs.~$\tilde{\beta}$ \eqref{eq:def_tilde_alpha} and $\tilde{f}$ \eqref{eq:def_tilde_f}}
    \end{align*}
    Lastly, the functoriality of $F$ is clear from its definition, and the proof is complete.
\end{proof}

\begin{definition}[Functor G]
    \label{def:G}
    We define the mapping
    \[
        G:
        \FuzzyPresheaf{\elements{D}}{\tilde{\labels}}
        \to
        \quot{\FuzzyPresheafDefault}{(D,\delta)}
    \]
    by $G(\tilde{A}, \tilde{\alpha}) \defeq (A,\alpha, p)$, where $A(i) \defeq \smash{\bigsqcup_{d \in D(i)}} \tilde{A}(i,d)$,
    \[
        A(\iota : j \to i) \defeq D(\iota) \times \tilde{A}(\iota) : \bigsqcup_{d \in D(i)} \tilde{A}(i,d) \to \bigsqcup_{e \in D(j)} \tilde{A}(j,e)
    \]
    \begin{equation}
        \begin{tikzcd}[row sep = tiny, ampersand replacement=\&]
            {A(i) = \bigsqcup_{d \in D(i)} \tilde{A}(i,d)}
                \ar[rd, color=gray, "{ \alpha_i \defeq (d,\tilde{a}) \mapsto \tilde{\alpha}_{(i,d)} (\tilde{a})}"']
                \ar[rr, "p_i \defeq \pi_1"]
            \& \& D(i)
                \ar[dl, color=gray, "\delta_i"]
            \\
            \& {\color{gray}\labels(i)} \&
        \end{tikzcd}
        \label{eq:def_alpha_p} 
    \end{equation}
    and $G(\tilde{f}) \defeq f$, where
    \begin{equation}
        f_i : 
        {\displaystyle \bigsqcup_{d \in D(i)} \tilde{A}(i,d)}
        \to 
        {\displaystyle \bigsqcup_{d \in D(i)} \tilde{B}(i,d)}
        :
        (d , \tilde{a})
        \mapsto
        (d, \tilde{f}_{(i,d)}(\tilde{a})).
        \label{eq:def_f}
    \end{equation}


\end{definition}

\begin{lemma}
    \label{lem:G_is_well_defined}
    $G$ is a well-defined functor.
\end{lemma}

\begin{proof}
    We first check that $p$ is natural in $i \in I$:
    \begin{center}\begin{tikzcd}[column sep=large]
        (d,\tilde{a})
            \ar[r, mapsto, "D(\iota) \times \tilde{A}(\iota)"]
            \ar[d, mapsto, "p_i = \pi_1"']
        & \big( D(\iota)(d), \tilde{A}(\iota)(\tilde{a}) \big)
            \ar[d, mapsto, "p_j = \pi_1"]
        \\
        d
            \ar[r, mapsto, "D(\iota)"']
        & D(\iota)(d)
    \end{tikzcd}\end{center}
    We then check that $\alpha \leq \delta p$.
    Take $i \in I$ and $(d,\tilde{a}) \in A(i)$.
    Since
    \[
        \alpha_i(d,\tilde{a}) \stackrel{\text{Def.~}\ref{def:G}}{=} \tilde{\alpha}_{(i,d)} (\tilde{a}) \in \tilde{\labels}(i,d) = \labels(i)_{\leq \delta_i(d)},
    \]
    it follows that, $\alpha_i(d,\tilde{a}) \leq \delta_i(d) = \delta_i p_i (d, \tilde{a})$, as desired.
    Moreover, $p = qf$ holds immediately from the definition.
    We also check that $\alpha \leq \beta f$: for $i \in I$ and $(d, \tilde{a}) \in A(i)$ we have
    \begin{align*}
        \alpha_i (d, \tilde{a})
        &= \tilde{\alpha}_{(i,d)} (\tilde{a})
        \tag{def.~$\alpha$ \eqref{eq:def_alpha_p}} \\
        &\leq \tilde{\beta}_{(i,d)} \big( \tilde{f}_{(i,d)} (\tilde{a}) \big)
        \tag{def.~$\tilde{f}$} \\
        &= \beta_i \big(d, \tilde{f}_{(i,d)} (\tilde{a}) \big)
        \tag{def.~$\beta$ \eqref{eq:def_alpha_p}} \\
        &= \beta_i f_i (d, \tilde{a}).
        \tag{def.~$f$ \eqref{eq:def_f}}
    \end{align*}
    Lastly, the functoriality of $G$ is clear from its definition, and the proof is complete.
\end{proof}

\begin{lemma}
    \label{lem:F_and_G_are_quasi_inverses}
    $F$ and $G$ from \cref{def:F,def:G} are quasi-inverses.
\end{lemma}

\begin{proof}
    We want to construct two natural isomorphisms.
    \begin{enumerate}
        \item
        The first one is $\sigma : GF \xlongrightarrow{\isom} \id{ \nicefrac{\scriptstyle \FuzzyPresheafDefault}{\scriptstyle (D,\delta)}}$.
        Given $(A,\alpha,p)$, let $(\tilde{A},\tilde{\alpha}) \defeq F(A,\alpha)$ and $(A', \alpha', p') \defeq G(\tilde{A}, \tilde{\alpha})$.
        For $i \in I$, let $\sigma_{(A,\alpha,p),i}$ and $\sigma_{(A,\alpha,p),i}\inv$ be the partitioning and de-partitioning of $A(i)$:
        \begin{align*}
            A'(i) = \bigsqcup_{d \in D(i)} p_i \inv (d) 
            \qquad &\stackrel[\mathclap{\sigma_{(A,\alpha,p),i} \inv}]{\mathclap{\sigma_{(A,\alpha,p),i}}}{\rightleftarrows} \qquad 
            A(i), \\
            (d,a) \qquad &\mapsto \qquad a, \\
            (p_i(a), a) \qquad &\mapsfrom \qquad a.
        \end{align*}
        Observe that indeed $p' = p  \sigma_{(A,\alpha,p)}$: for $i \in I$ and $(d \in D(i), a \in p_i \inv (d))$, we have
        \[
            p_i ( \sigma_{(A,\alpha,p),i} (d,a)) = p_i (a) = d = \pi (d,a) = p_i' (d,a).
        \]
        We also have $\alpha' \leq \alpha \sigma_{(A,\alpha,p)}$, as we even have equality: for $i \in I$ and $(d \in D(i), a \in p_i \inv (d))$, we have:
        \[
            \alpha'_i (d,a) = \alpha_i (a) = \alpha_i \sigma_{(A,\alpha,p),i}(d,a).
        \]
        We have $\sigma$ natural in $(A,\alpha,p)$: given $f:(A,\alpha,p) \to (B,\beta,q)$
        \[
            \begin{tikzcd}
            	{GF(A,\alpha,p)} & {(A,\alpha,p)} \\
            	{GF(B,\beta,q)} & {(B,\beta,q)}
            	\arrow["{GF(f)}"', from=1-1, to=2-1]
            	\arrow["f", from=1-2, to=2-2]
            	\arrow["{\sigma_{(A,\alpha,p)}}", from=1-1, to=1-2]
            	\arrow["{\sigma_{(B,\beta,q)}}"', from=2-1, to=2-2]
            \end{tikzcd}
            \quad
            \stackrel[\text{with $i$}]{\text{instantiated}}{\Rightarrow}
            \quad
            \begin{tikzcd}
            	{\bigsqcup_{d \in D(i)} p_i\inv(d)} & {A(i)} \\
            	{\bigsqcup_{d \in D(i)} q_i\inv(d)} & {B(i)}
            	\arrow["{\bigsqcup_{d \in D(i)} \inl_d f_i}"', from=1-1, to=2-1]
            	\arrow["{f_i}", from=1-2, to=2-2]
            	\arrow["{\pi_2}", from=1-1, to=1-2]
            	\arrow["{\pi_2}"', from=2-1, to=2-2]
            	\arrow["\circlearrowleft", phantom, from=1-1, to=2-2]
            \end{tikzcd}
        \]
        and $\sigma_{(A,\alpha,p)}$ natural in $i$: given $j \iota i$
        \[\begin{tikzcd}
        	{\bigsqcup_{d \in D(i)} p_i\inv(d)} & {A(i)} \\
        	{\bigsqcup_{e \in D(j)} p_j\inv(e)} & {A(j)}
        	\arrow["{D(\iota) \times A(\iota)}"', from=1-1, to=2-1]
        	\arrow["{A(\iota)}", from=1-2, to=2-2]
        	\arrow["{\pi_2}", from=1-1, to=1-2]
        	\arrow["{\pi_2}"', from=2-1, to=2-2]
        	\arrow["\circlearrowleft", phantom, from=1-1, to=2-2]
        \end{tikzcd}\]
        
        \item
        The second one is $\tau : FG \xlongrightarrow{\isom} \id{(\Presheaf{\elements{D}}, \tilde{\labels})}$.
        Given $(\tilde{A},\tilde{\alpha})$, let $(A,\alpha,p) \defeq G(\tilde{A},\tilde{\alpha})$ and $(\tilde{A}', \tilde{\alpha}')$.
        For $(i,d) \in \opcat{\elements{D}}$:
        \begin{align*}
            (\tilde{A}', \tilde{\alpha}') (i,d)
            &=  \Big( p_i\inv(d) \xrightarrow{\alpha_i} \labels(i)_{\leq \delta_i(d)} \Big)
            \tag{def.~$F$} \\
            &= \Big( \set{d} \times \tilde{A}(i,d) \xrightarrow{\tilde{\alpha}_{(i,d)} \cdot \pi_2} \tilde{\labels(i,d)} \Big)
            \tag{def.~$G$}
        \end{align*}
        So let 
        \[
            \begin{tikzcd}[column sep = large]
                \set{d} \times \tilde{A}(i,d)
                && \tilde{A}(i,d) \\
                & {\color{gray}\tilde{\labels}(i,d)}
                \arrow["\tau_{ (\tilde{A}, \tilde{\alpha}), (i,d) } \defeq \pi_2", shift left, from=1-1, to=1-3]
                \arrow["\tau_{ (\tilde{A}, \tilde{\alpha}), (i,d) }\inv \defeq \set{d} \times \id{ \tilde{A}(i,d) }", shift left, from=1-3, to=1-1]
                \arrow["\tilde{\alpha}_{(i,d)} \cdot \pi_2"', color=gray, from=1-1, to=2-2]
                \arrow["\tilde{\alpha}_{(i,d)}", color=gray, from=1-3, to=2-2]
            \end{tikzcd}
        \]
        We have $\tau$ natural in $(\tilde{A}, \tilde{\alpha}):$ given $\tilde{f} : (\tilde{A}, \tilde{\alpha}) \to (\tilde{B}, \tilde{\beta})$
        \[
            \begin{tikzcd}
            	{FG(\tilde{A},\tilde{\alpha})} & {(\tilde{A},\tilde{\alpha})} \\
            	{FG(\tilde{B},\tilde{\beta})} & {(\tilde{B},\tilde{\beta})}
            	\arrow["{FG(\tilde{f})}"', from=1-1, to=2-1]
            	\arrow["\tilde{f}", from=1-2, to=2-2]
            	\arrow["{\tau_{(\tilde{A},\tilde{\alpha})}}", from=1-1, to=1-2]
            	\arrow["{\tau_{(\tilde{B},\tilde{\beta})}}"', from=2-1, to=2-2]
            \end{tikzcd}
            \quad
            \stackrel[\text{with $(i,d)$}]{\text{instantiated}}{\Rightarrow}
            \quad
            \begin{tikzcd}
            	{\set{d} \times \tilde{A}(i,d)} & {\tilde{A}(i,d)} \\
            	{\set{d} \times \tilde{B}(i,d)} & {\tilde{B}(i,d)}
            	\arrow["{\id{\set{d}} \times \tilde{f}_{(i,d)}}"', from=1-1, to=2-1]
            	\arrow["{\tilde{f}_{(i,d)}}", from=1-2, to=2-2]
            	\arrow["{\pi_2}", from=1-1, to=1-2]
            	\arrow["{\pi_2}"', from=2-1, to=2-2]
            	\arrow["\circlearrowleft", phantom, from=1-1, to=2-2]
            \end{tikzcd}
        \]
        and $\tau_{(\tilde{A},\tilde{\alpha})}$ natural in $(i,d)$: given $(j,e) \xrightarrow{\iota} (i,d)$
        \[\begin{tikzcd}
        	{\set{d} \times \tilde{A}(i,d)} & {\tilde{A}(i,d)} \\
        	{\set{e} \times \tilde{A}(j,e)} & {\tilde{A}(j,e)}
        	\arrow["{D(\iota) \times \tilde{A}(\iota)}"', from=1-1, to=2-1]
        	\arrow["{\tilde{A}(\iota)}", from=1-2, to=2-2]
        	\arrow["{\pi_2}", from=1-1, to=1-2]
        	\arrow["{\pi_2}"', from=2-1, to=2-2]
        	\arrow["\circlearrowleft", phantom, from=1-1, to=2-2]
        \end{tikzcd}\]
    \end{enumerate}
\end{proof}

\subsection{Proof of \cref{sec:simple_fuzzy_graphs}}

\begin{proof}[Proof of \cref{lem:details_notnot_topology}]
    We do the reasoning for each category.
    \begin{itemize}
        \item 
        In $\Set$, given a set $A$ the morphism from the initial object $0_A : \emptyset \mono A$ is the empty function.
        It factors as $\overline{0_A} \cdot e: \emptyset \epimono \emptyset \regmono A$, hence $\overline{0_A}$ is still the empty function $\emptyset$.
        Using the formula for exponential objects in the slice category $\Set/A$ \cite[p.~81]{Stout_1993}, we have for $m : A_0 \mono A$:
        \begin{align*}
            \lnot m
            &= (\emptyset: \emptyset \mono A)^{(m : A_0 \mono A)}
            \tag{def.~$\lnot$} \\
            &= \pi_1 : \setvbar{(a,h)}{a \in A, h : m\inv(a) \to \emptyset\inv(a)} \mono A
            \tag{formula} \\
            &= \inclusion : \setvbar{a \in A}{m \inv (a) = \emptyset} \mono A.
            \tag{$h$ must be $\emptyset$}
        \end{align*}
        When $A_0 \subseteq A$, then $\lnot A_0 = \set{a \notin A} = A \setminus A_0$, as desired.

        \item
        In $\FuzzySetDefault$, the initial object is $(\emptyset, \emptyset \xrightarrow{\emptyset} \mathcal{L})$.
        Given a fuzzy set $(A,\alpha)$, the factorisation of the morphism $0_A$ from the initial object does not change anything, as in $\Set$, and thus $0_A = \overline{0_A} = (\emptyset, \emptyset) \mono (A, \alpha)$.
        By definition, 
        \[
            \lnot \big((A_0,\alpha_0) \stackrel{m}{\mono} (A,\alpha) \big) = \big( (\emptyset,\emptyset) \stackrel{\emptyset}{\mono} (A,\alpha) \big)^{\big((A_0,\alpha_0) \stackrel{m}{\mono} (A,\alpha)\big)}.
        \]
        Using the formula for exponential objects in the slice category $\FuzzySetDefault/(A,\alpha)$ \cite[p.~81]{Stout_1993}, we obtain the same carrier set as in $\Set$, and the membership function is 
        \begin{align*}
            \xi(a,h) &= \alpha(a) \meet \textstyle{\bigmeet_{a' \in m\inv(a)} \big(\alpha(a') \Rightarrow \emptyset h (a') \big)} 
            \tag{formula} \\
            &= \alpha(a) \meet \top
            \tag{big meet is empty, so by convention is $\top$} \\
            &= \alpha(a).
        \end{align*}
        \[
        \begin{tikzcd}
            \setvbar{(a,h)}{h : m\inv(a) \to \emptyset\inv(a)} & & A \\
            & \labels
            \ar[from=1-1, to=1-3, "\pi_1"]
            \ar[from=1-1, to=2-2, "\xi"']
            \ar[from=1-3, to=2-2, "\alpha"]
        \end{tikzcd}
        \]
        We can simplify the carrier set, as in $\Set$.
        When $A_0 \subseteq A$, we obtain thus $\lnot (A_0,\alpha_0) = (A \setminus A_0, \alpha)$, as desired.

        \item
        In $\Graph$, given a graph $A$, we similarly have $0_A = \overline{0_A} = \emptyset \mono A$.
        Take $m : A_0 \mono A$.
        We have
        \begin{align*}
            \lnot m
            &= (\emptyset \stackrel{\emptyset}{\mono} A)^{(A_0 \stackrel{m}{\mono} A)}
            \tag{def.~$\lnot$} \\
            &=
            \begin{cases}
                \text{vertex set:} \setvbar{(v,h)}{v \in A(V), h:\set{\cdot} \times m\inv(v) \to \emptyset\inv(v)}  \\
                \text{edge set:} \setvbar{(e,k)}{e \in A(E), k:\set{\cdot \to \cdot} \times m\inv(e) \to \emptyset\inv(e)}
            \end{cases}
            \stackrel{\inclusion}{\mono} A
            \tag{formula} \\
            &=
            \begin{cases}
                \text{vertex set:} \setvbar{v \in A(V)}{m\inv(v) = \emptyset} \\
                \text{edge set:} \setvbar{e \in A(E)}{m\inv(e) = \emptyset} 
            \end{cases}
            \stackrel{\inclusion}{\mono} A
            \tag{simplifying}
        \end{align*}
        Notice that this gives the largest subgraph of $A$ totally disconnected from $A_0$.
        Indeed, consider $A_0 \subseteq A$ for simplification.
        Then $m \inv(v) = \emptyset \iff \set{v} \cap A_0 = \emptyset$, i.e., $v \notin A_0(V)$.
        However, $m \inv(e) \neq \emptyset \iff \set{e} \cap A_0 = \emptyset$, does \textit{not} mean that $e \notin A_0$.
        Its source or target can be in $A_0$ and thus in the intersection.
        Hence,
        \[
            \lnot \lnot m = 
            \begin{cases}
                \text{vertex set:} \setvbar{v \in A(V)}{m\inv(v) \neq \emptyset} \\
                \text{edge set:} \setvbar{e \in A(E)}{m\inv(e) \neq \emptyset} \ \ 
            \end{cases}
            \stackrel{\inclusion}{\mono} A,
        \]
        is the subgraph $A_0$ where all edges with source and target in $A_0$ are added.
        
        \item
        In $\FuzzyGraph$, given $(A_0,\alpha_0) \subseteq (A,\alpha)$, the $\lnot \lnot$ operation is a combination of both previous cases.
        On the underlying subgraph $A_0$, it gives the largest subgraph of $A$ totally disconnected from $A_0$.
        On the membership function, it replaces it with the one $\alpha$ of the main object.
        \qedhere
    \end{itemize}
\end{proof}

\begin{proof}[Proof of \cref{lem:separated_elements_in_fuzzy_multigraphs_are_fuzzy_simple_graphs}]
    The reasoning is the same as in \cref{lem:separated_elements_in_multigraphs_are_simple_graphs}.
    One might think that that the membership functions is problematic, for instance because of situations like that:
    \[
        \begin{tikzcd}[ampersand replacement=\&, sep=small]
            \set{\cdot^{0.2}} \& \\
            \set{\cdot^{0.6}} \& \set{\cdot^{0.4}}
            \ar[from=1-1, to=2-1, tail]
            \ar[from=1-1, to=2-2]
            \ar[from=2-1, to=2-2, dotted, "\nexists"']
        \end{tikzcd}
    \]
    However, the definition of separated elements in \cref{def:dense_subobjects_and_separated_elements} says that \textit{for all} dense subobjects $B_0 \mono B$, there must exist at most one factorisation.
    Hence, it is always possible to consider the case where the membership values in $B_0$ and $B$ are the same, making it the same situation as \cref{lem:separated_elements_in_multigraphs_are_simple_graphs}, and hence the same reasoning.
\end{proof}

\end{document}